\documentclass[journal, twocolumn,final]{IEEEtran}
\usepackage{comment}
\usepackage{amsfonts}
\usepackage{amssymb}
\usepackage{amsmath}
\DeclareMathOperator{\Tr}{Tr}
\usepackage{dsfont}
\usepackage{graphicx}
\usepackage{bm}
\usepackage{cite}
\usepackage{bigstrut}
\usepackage{float}
\usepackage{balance}
\usepackage{lipsum}
\usepackage{psfrag}
\usepackage{xcolor}
\usepackage{xfrac}
\usepackage{hyperref}

\usepackage{algorithm}
\usepackage{algpseudocode}

\usepackage[caption=false,font=footnotesize]{subfig}

\usepackage{tikz}
\usepackage{pgfplots}
\pgfplotsset{compat=1.18}
\usepackage{xcolor}
\usepackage{mathrsfs}
\usetikzlibrary{arrows}
\usetikzlibrary{decorations.pathreplacing}
\usetikzlibrary{fadings}
\usetikzlibrary{fit}

\newtheorem{lemma}{Lemma}

\newtheorem{proposition}{Proposition}

\usepackage{mathtools}
\usepackage{comment}
\usepackage{tabularx,array}

\begin{document}
\IEEEoverridecommandlockouts
\newcommand\norm[1]{\left\lVert#1\right\rVert}
\pgfplotsset{compat=1.17}
\usepgfplotslibrary{patchplots}   

\makeatletter
\long\def\@makecaption#1#2{%
    \ifx\@captype\@IEEEtablestring%
        \footnotesize
        \bgroup
        \par
        \centering
        \@IEEEtabletopskipstrut
        {\normalfont\footnotesize #1}\\
        {\normalfont\footnotesize\scshape #2}
        \par
        \addvspace{0.5\baselineskip}
        \egroup
        \@IEEEtablecaptionsepspace
    \else
        \@IEEEfigurecaptionsepspace
        \parbox[t]{\hsize}{%
            \normalfont\footnotesize
            \centering
            {#1.}\nobreakspace\nobreakspace #2\par
        }%
    \fi
}
\makeatother

\definecolor{myblue}{RGB}{0,114,178}
\definecolor{myorange}{RGB}{213,94,0}
\definecolor{mygreen}{RGB}{0,158,115}
\definecolor{mypurple}{RGB}{117,112,179}

\title{Secure ISAC with Sensing Privacy under
Eavesdropper Uncertainty}
\author{Pigi P. Papanikolaou,~\IEEEmembership{Graduate Student Member,~IEEE}, Dimitrios Bozanis,~\IEEEmembership{Graduate Student Member,~IEEE}, Sotiris A. Tegos,~\IEEEmembership{Senior Member,~IEEE}, Christos Masouros,~\IEEEmembership{Fellow,~IEEE},\\ and George K. Karagiannidis,~\IEEEmembership{Fellow,~IEEE}

\thanks{P. P. Papanikolaou, D. Bozanis, S. A. Tegos and G. K. Karagiannidis are with the Department of Electrical and Computer Engineering, Aristotle University of Thessaloniki, 54124 Thessaloniki, Greece (e-mails: pigipapa@auth.gr, dimimpoz@ece.auth.gr, tegosoti@auth.gr, geokarag@auth.gr).}
\thanks{C. Masouros is with the Department of Electronic and Electrical Engineering, University College London, London, UK (e-mail: c.masouros@ucl.ac.uk).}
\thanks{}
\vspace{-5mm}

}

\maketitle

\begin{abstract}
This paper investigates the joint protection of confidential data and legitimate-user directional information in integrated sensing and communication (ISAC) networks. We consider a multiuser downlink in which passive multi-antenna eavesdroppers (Eves) attempt to decode confidential signals while exploiting legitimate-user reflections for unauthorized angular sensing. To address both threats, we develop a robust covariance-design framework that jointly limits Eve decoding, maintains the transmitter’s accuracy in estimating the Eves’ directions and shifts the dominant passive-sensing response toward prescribed deceptive directions. A Bayesian angular prior and the corresponding Bayesian Cram\'er--Rao bound (BCRB) characterize the transmitter's eavesdropper-angle estimation accuracy. Angular uncertainty is represented through geometry-consistent samples, such that each candidate Eve direction jointly determines the corresponding transmitter--Eve channel, user--Eve bearing, and deceptive direction. The resulting design balances worst-user secrecy, sensing accuracy, sensing privacy, and deception power. Robust Eve-decoding constraints are handled through finite sufficient conditions with intersample margins, while continuous ghost dominance is enforced using interval sum-of-squares (SOS) constraints. The resulting nonconvex problem is addressed through successive convex approximation (SCA) and semidefinite relaxation. Numerical results show that the proposed design effectively preserves secrecy and sensing privacy under Eve-angle uncertainty, provides controlled angular deception in single- and multiple-Eve scenarios, and outperforms benchmarks.
\end{abstract}


    
\section{Introduction}
\label{sec:intro}

Integrated sensing and communication (ISAC) is regarded a key technology for future wireless networks because it enables sensing and communication functions to share spectrum, hardware, antenna arrays and transmitted waveforms \cite{Zhang2021JCR,Liu2022ISAC}. Building on earlier radar--communication coexistence studies \cite{Zheng2019Coexistence}, ISAC has evolved toward the joint design of dual-functional transmit signals and processing architectures. This integration can improve spectral efficiency, reduce deployment costs, and support applications such as vehicular networks, smart manufacturing, environmental monitoring, and ubiquitous Internet-of-Things (IoT) services \cite{Cui2021UbiquitousIoT}.

Beyond these benefits, the utilization of a shared waveform tightly couples the communication and sensing functionalities. The transmitted signal must deliver information to legitimate users while illuminating targets or regions of interest for environmental sensing. Reliable sensing generally requires sufficient energy to be directed toward these directions, for example, to achieve a desired transmit beampattern or to reduce the Cram\'er--Rao bound of the estimated parameters \cite{18,Boz_crb}. However, when the illuminating waveform also carries confidential information, the resulting signal exposure may enable a target or another unauthorized receiver to intercept and decode the transmitted data \cite{sec2,sec3}.

Yet, the information exposed by the shared waveform is not limited to the transmitted data. An unauthorized receiver may also process the direct and scattered signal components to infer information about legitimate users or the surrounding environment. ISAC systems therefore face two distinct but coupled security concerns: communication confidentiality, which aims to prevent unauthorized data decoding, and sensing privacy, which aims to restrict the inference of spatial information \cite{Wei2022Multifunctional}. However, even when unauthorized data decoding is effectively suppressed, an eavesdropper may still extract meaningful information about legitimate users or the surrounding environment from its received observations, leaving the issue of sensing privacy unresolved.

Accordingly, security in ISAC must extend beyond the protection of confidential data to also account for the sensing information exposed by the shared waveform. Communication confidentiality and sensing privacy therefore emerge as complementary dimensions of ISAC protection, motivating the development of transmission strategies that address both unauthorized decoding and unauthorized sensing.

\subsection{State of the Art}

Research on secure ISAC systems has focused primarily on protecting communication confidentiality. Secrecy rate (SR) optimization for joint radar and communication systems was considered in
\cite{Deligiannis2018SecrecyRadar}, while secure dual-functional radar--communication designs in which sensing targets may also act as information eavesdroppers were developed in
\cite{Su2021MaliciousTargets,Su2022SecureDFRC}. These approaches employ information beamforming, dedicated radar signals, artificial noise (AN), or multiuser (MU) interference to preserve the communication quality of legitimate users while limiting the targets' decoding capabilities. Moving beyond single-snapshot designs, \cite{Xu2022RobustSecureISAC} developed a robust resource-allocation framework that jointly optimizes the snapshot durations, information beamformers, and AN covariance matrices over a sequence of variable-length snapshots while accounting for imperfect CSI at both legitimate users and potential eavesdroppers.

The sensing functionality of ISAC has also been exploited to improve communication security by providing information about potential eavesdroppers. In \cite{Su2024SensingAssisted}, the transmitter first estimates the angular directions of potential eavesdroppers and subsequently accounts for the estimation errors through a widened transmit beampattern, thereby enforcing communication security over the possible eavesdropper region while jointly optimizing the SR and sensing accuracy. Sensing information has similarly been used to track and suppress mobile aerial eavesdroppers
\cite{Liu2023AerialEve} and to enable sensing-aided covert transmission under perfect and imperfect knowledge of the adversary's channel
\cite{Wang2024SensingAidedCovert}. Bayesian secure ISAC designs have further incorporated prior target-location distributions and Bayesian estimation bounds to account for sensing uncertainty in waveform optimization
\cite{Hou2023TargetDistribution,Su2024BCRB}. While these approaches account for uncertainty in the parameter estimation, \cite{Boz_iot} further incorporates it into security by enforcing destructive interference over the entire eavesdropper uncertainty region. Nevertheless, in these studies, sensing primarily assists the legitimate transmitter in localizing potential eavesdroppers or enforcing communication security. The sensing information that an unauthorized receiver may independently infer from the transmitted waveform is not explicitly protected.

Therefore, sensing privacy has recently emerged as a distinct aspect of ISAC security. In
\cite{Zou2024SensingSecurity}, AN-aided beamforming is used to improve the mutual information available to a legitimate sensing receiver while limiting the sensing mutual information at an unauthorized receiver. Information and sensing eavesdroppers are jointly considered in a cell-free ISAC architecture in
\cite{Ren2024CellFreeSecure}. Such approaches aim to degrade the unauthorized receiver's detection, estimation, or overall sensing performance. However, suppressing the information available to a sensing eavesdropper generally produces an uncontrolled estimation error. In particular, it does not determine which incorrect location or direction will emerge as the dominant output of the eavesdropper's sensing processor.

To go beyond sensing degradation, more recent studies have investigated controlled sensing deception. Artificial imperfections in the ambiguity function have been introduced to generate false peaks in an unauthorized receiver's range profile
\cite{Han2025SensingSecure}. Near-field scatterers have also been deliberately illuminated to produce misleading sensing observations
\cite{Chen2025NearFieldDeception}, while a MIMO-OFDM ISAC framework for generating artificial targets with prescribed angle, range, and velocity was developed in
\cite{Yang2026DualSecurity}. These studies demonstrate the potential of actively controlling the sensing information perceived by an unauthorized receiver rather than merely reducing its sensing quality.

\subsection{Motivation and Contributions}

Despite these advances, an important gap remains. Existing sensing privacy and controlled deception methods either suppress unauthorized sensing without explicitly controlling the resulting estimate or generate misleading sensing responses through ambiguity-function shaping, known scatterers, or receiver-independent waveform design. However, they do not explicitly model uncertainty in the geometry of the unauthorized sensing receiver and consistently propagate its effects across the communication and sensing models. In particular, such uncertainty induces coupled variations in the eavesdropper channel, the observed target or user directions, and the prescribed deceptive response. Treating these quantities independently may therefore lead to physically inconsistent and non-robust designs.

To address this gap, this paper develops a secure ISAC framework that jointly protects communication confidentiality and sensing privacy under imperfect knowledge of the eavesdropper geometry. We consider an MU downlink in which a multi-antenna transmitter serves multiple single-antenna legitimate users while sensing the directions of passive multi-antenna eavesdroppers. Although the legitimate users are not intended sensing targets of the transmitter, their reflected signal components may be exploited by the eavesdroppers to perform unauthorized angular sensing and infer the users' directions. The proposed design therefore restricts the eavesdroppers' data-decoding capability while shaping their passive sensing observations so that the dominant response appears near a prescribed virtual direction rather than the true direction of a legitimate user. To the best of our knowledge, this is the first secure MU ISAC design to jointly address communication confidentiality and controlled sensing privacy under uncertain eavesdropper geometry.


The considered threat focuses specifically on the spatial-information leakage created by the downlink waveform. During the considered downlink interval, the legitimate users operate as receivers and no exploitable uplink transmission is assumed. Hence, a passive sensing-capable Eve may exploit the user-reflected downlink component to infer a legitimate user's direction, as considered in \cite{Musallam}. 
Each Bob represents the physical communication user or platform associated with the communication receiver, whose scattering response is characterized through an effective user RCS. 

The main contributions of this paper are summarized as follows:

\begin{itemize}

    \item We develop a secure MU-ISAC framework in which passive multi-antenna eavesdroppers threaten both communication confidentiality and sensing privacy. Besides attempting to decode the confidential downlink signals, each eavesdropper exploits the reflection from a selected legitimate user to infer that user's direction through passive sensing. Rather than merely degrading this sensing capability, the proposed framework controls the resulting estimate by inducing a dominant response at a prescribed virtual direction.

    \item We derive the Bayesian Cramér-Rao Bound (BCRB) for Alice’s estimation of each eavesdropper direction and use it as a sensing metric suitable for incorporating prior information and directional uncertainty. The resulting uncertainty sectors are accounted for through a geometry-consistent model, where each admissible eavesdropper direction jointly determines the corresponding eavesdropper channel, the legitimate-user directions observed in the passive scan, and the prescribed virtual direction.

    \item We formulate a robust optimization problem that jointly designs the information and deception covariance matrices. The objective maximizes the worst-user secrecy margin and the separation between the prescribed virtual response and competing passive-scan responses, while minimizing the BCRB for accurate sensing and the deception-power fraction. The design further enforces legitimate-user QoS and eavesdropper-decoding constraints over the entire uncertainty region, while a conservative intersample margin is introduced to account for variations between adjacent uncertainty samples.

    \item We develop a tractable successive convex approximation (SCA)-based solution for the resulting nonconvex joint covariance design. The rank-one beamforming constraints are relaxed, the continuous Eve-decoding requirement is replaced by finite sufficient constraints with intersample robustness margins, the continuous ghost-dominance requirement is reformulated through interval sum-of-squares (SOS) constraints, while tight concave lower bounds are constructed for the legitimate-user rates at each iteration, yielding a convex semidefinite-constrained subproblem. Proximal regularization is incorporated to stabilize the covariance updates, and beamforming vectors are recovered from the optimized covariance matrices. 


\end{itemize}

Moreover, numerical results validate the proposed framework and characterize its secrecy, sensing-privacy, robustness, objective-weight tradeoffs, convergence, and MU scalability.

\subsection{Structure}

The remainder of this paper is organized as follows. Section~\ref{sec:SysMod} presents the system, channel, communication, and passive-sensing models. Section~\ref{sec:BayesianUncertainty} presents the BCRB-based uncertainty model and the resulting angular sectors for the estimated eavesdropper directions. The joint communication-security and sensing-deception design is formulated in Section~\ref{sec:Optsol}, where the proposed SCA-based solution is also presented. Numerical results are reported in Section~\ref{sec:Num}, and Section~\ref{sec:Conc} concludes the paper.

\section{System Model}\label{sec:SysMod}

\begin{figure}
    \centering
    \includegraphics[width=0.9\columnwidth]{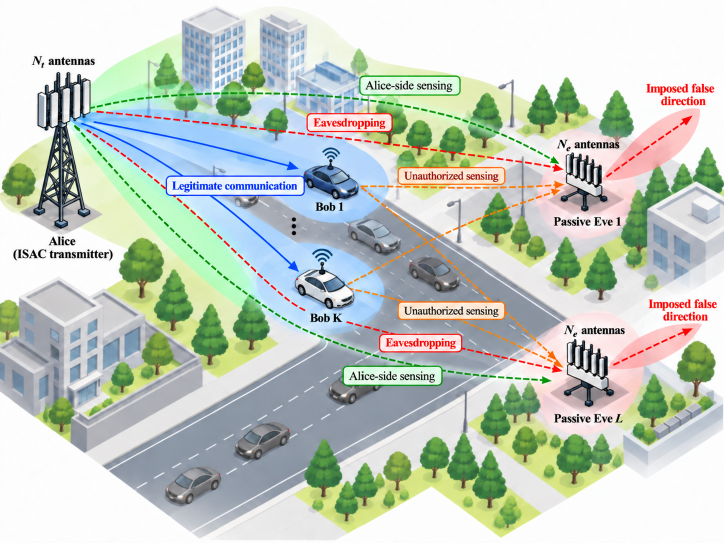}
    \vspace{-2mm}
    \caption{Illustration of the system model.}
    \vspace{-4mm}
    \label{fig:sys_model}
\end{figure}

We consider a secure downlink ISAC system composed of one multi-antenna transmitter, referred to as Alice, $K$ legitimate users, referred to as Bobs, and $L$ passive eavesdroppers, referred to as Eves, whose directions are sensed by Alice for security purposes, 
as shown in Fig.~\ref{fig:sys_model}. Alice serves the Bobs with confidential downlink communication while sensing the directions of the passive Eves. 
Alice is equipped with a uniform linear array (ULA) with $N_t$ antennas. Each Bob is a single-antenna communication receiver, while each Eve employs a ULA with $N_e$ antennas. 
We assume that the ULAs at Alice and the Eves are parallel, with their zero-angle reference aligned with the global $x$-axis. Accordingly, all azimuth angles are measured counterclockwise from the global $x$-axis and can be used directly in the corresponding steering vectors. The node positions are denoted by 
\begin{equation}
\begin{aligned}
&\mathbf p_A=[x_A,y_A]^T\in\mathbb R^2, \mathbf p_{B,k}=[x_{B,k},y_{B,k}]^T\in\mathbb R^2, \\
&\mathbf p_{E,\ell}=[x_{E,\ell},y_{E,\ell}]^T\in\mathbb R^2,
\end{aligned}
\end{equation}
for $k\in\{1,\ldots,K\}$ and $\ell\in\{1,\ldots,L\}$, respectively. As associated legitimate users, the Bobs' positions and downlink CSI are assumed known at Alice.


Alice communicates with the Bobs while sensing the passive Eves.
Although the Bobs are not intended sensing targets of Alice, their
downlink-induced reflections may be exploited by the Eves for
unauthorized angular sensing. Hence, the system must protect the
confidential downlink communication while controlling the angular
information that passive Eves can infer from the downlink waveform.

\subsection{Transmit Signal Model}

Alice transmits information-bearing signals to the Bobs
together with a deception signal designed to limit the Eves'
ability to decode confidential information and shape their
passive-sensing responses. Accordingly, the transmit signal is
given by
\begin{equation}
    \mathbf x
    =
    \sum_{k=1}^{K}\mathbf w_k s_k
    +
    \mathbf z,
\end{equation}
where $s_k\sim\mathcal{CN}(0,1)$ denotes the information symbol
intended for the $k$-th Bob, $\mathbf w_k\in\mathbb C^{N_t}$ is
the corresponding beamforming vector, and
$\mathbf z\sim\mathcal{CN}(\mathbf 0,\mathbf Z)$ denotes the
deception signal. The information symbols and the deception
signal are assumed to be mutually independent.

The covariance matrix of the information signal intended for the
$k$-th Bob is
\begin{equation}
    \mathbf W_k
    =
    \mathbf w_k\mathbf w_k^H
    \succeq \mathbf 0,
    \qquad
    \operatorname{rank}(\mathbf W_k)=1.
\end{equation}
The covariance matrix of the deception signal is
\begin{equation}
    \mathbf Z
    =
    \mathbb E
    \left\{
        \mathbf z\mathbf z^H
    \right\}
    \succeq
    \mathbf 0.
\end{equation}
Consequently, the total transmit covariance is
\begin{equation}
    \mathbf R_x
    =
    \mathbb E\{\mathbf x\mathbf x^H\}
    =
    \sum_{k=1}^{K}\mathbf W_k
    +
    \mathbf Z.
\end{equation}
The covariance $\mathbf R_x$ characterizes the overall spatial
distribution of the transmitted power and enters the legitimate
communication, eavesdropper decoding, transmitter-side sensing, and
Eves' passive-sensing models.

\subsection{Communication Signal Model}

The Alice--Bob channel for the $k$-th Bob is denoted by $\mathbf h_{B,k}\in\mathbb C^{N_t}$. Since Bob $k$ has a single antenna, its received signal is scalar and is given by
\begin{equation}
    y_{B,k}
    =
    \mathbf h_{B,k}^H\mathbf x
    +
    n_{B,k},
    \qquad
    n_{B,k}\sim\mathcal{CN}(0,\sigma_b^2),
\end{equation}
where $n_{B,k}$ denotes the additive white Gaussian noise at Bob $k$, with variance $\sigma_b^2$. For compact notation, we define the rank-one channel matrix
\begin{equation}
    \mathbf H_{B,k}
    =
    \mathbf h_{B,k}\mathbf h_{B,k}^H
    \succeq
    \mathbf 0.
\end{equation}
Since the $k$-th Bob decodes only its intended information stream, the
signals intended for the other Bobs and the deception signals are
treated as interference. Accordingly, we define its desired
signal power as
\begin{equation}
    S_{B,k}
    =
    \operatorname{\Tr}
    \left(
        \mathbf H_{B,k}\mathbf W_k
    \right),
\end{equation}
and its interference-plus-noise power as
\begin{equation}
    I_{B,k}
    =
    \sum_{\substack{j=1\\j\neq k}}^{K}
    \operatorname{\Tr}
    \left(
        \mathbf H_{B,k}\mathbf W_j
    \right)
    +
    \operatorname{\Tr}
    \left(
        \mathbf H_{B,k}\mathbf Z
    \right)
    +
    \sigma_b^2,
\end{equation}
where $j$ indexes the information streams intended for the other Bobs,
and $\operatorname{\Tr}(\cdot)$ denotes the trace operator. The
received SINR at Bob $k$ is therefore
\begin{equation}
    \mathrm{SINR}_{B,k}
    =
    \frac{S_{B,k}}{I_{B,k}},
\end{equation}
and its corresponding achievable rate is
\begin{equation}
    R_{B,k}
    =
    \log_2
    \left(
        1+\mathrm{SINR}_{B,k}
    \right).
\end{equation}

\subsection{Data Eavesdropping Model}

The $\ell$-th Eve receives Alice's downlink signal through the
physical Alice--Eve channel
$\mathbf H_{AE,\ell}\in\mathbb C^{N_e\times N_t}$. Thus, its received
signal is
\begin{equation}
    \mathbf y_{E,\ell}
    =
    \mathbf H_{AE,\ell}\mathbf x
    +
    \mathbf n_{E,\ell},
    \qquad
    \mathbf n_{E,\ell}
    \sim
    \mathcal{CN}
    \left(
        \mathbf 0,
        \sigma_e^2\mathbf I_{N_e}
    \right),
\end{equation}
where $\mathbf n_{E,\ell}$ is the additive white Gaussian noise
vector.

When the $\ell$-th Eve attempts to decode the information stream
intended for the $k$-th Bob, the transmit covariance of all remaining
signals is
\begin{equation}
    \mathbf X_k
    =
    \sum_{\substack{j=1\\j\neq k}}^{K}
    \mathbf W_j
    +
    \mathbf Z.
\end{equation}
The corresponding useful-signal and interference-plus-noise
covariance matrices are
\begin{equation}
    \mathbf R_{E,\ell k}^{\rm sig}
    =
    \mathbf H_{AE,\ell}
    \mathbf W_k
    \mathbf H_{AE,\ell}^H,
\end{equation}
and
\begin{equation}
    \mathbf R_{E,\ell k}^{\rm int}
    =
    \mathbf H_{AE,\ell}
    \mathbf X_k
    \mathbf H_{AE,\ell}^H
    +
    \sigma_e^2\mathbf I_{N_e},
\end{equation}
respectively.

Since
$\mathbf R_{E,\ell k}^{\rm int}\succ\mathbf 0$, the maximum
post-combining SINR available to the $\ell$-th Eve for decoding the
information stream intended for the $k$-th Bob is
\begin{equation}
    \gamma_{E,\ell k}
    =
    \lambda_{\max}
    \left(
        \left(
            \mathbf R_{E,\ell k}^{\rm int}
        \right)^{-1/2}
        \mathbf R_{E,\ell k}^{\rm sig}
        \left(
            \mathbf R_{E,\ell k}^{\rm int}
        \right)^{-1/2}
    \right),
    \label{eq:eve_post_combining_sinr}
\end{equation}
where $\lambda_{\max}(\cdot)$ denotes the largest eigenvalue of its
matrix argument. The inverse square root in
\eqref{eq:eve_post_combining_sinr} is well defined because
$\mathbf R_{E,\ell k}^{\rm int}$ is positive definite. Accordingly,
the achievable decoding rate of the $\ell$-th Eve is
\begin{equation}
    R_{E,\ell k}
    =
    \log_2
    \left(
        1+\gamma_{E,\ell k}
    \right).
    \label{eq:eve_decoding_rate}
\end{equation}

The SR associated with the $k$-th Bob is defined with
respect to the Eve having the strongest data-decoding capability as
\begin{equation}
    R_{{\rm sec},k}
    =
    \left[
        R_{B,k}
        -
        \max_{\ell\in\{1,\ldots,L\}}
        R_{E,\ell k}
    \right]^+,
    \label{eq:bob_secrecy_rate}
\end{equation}
where $[x]^+\triangleq\max\{x,0\}$. The corresponding worst-user SR is
\begin{equation}
    R_{\rm sec}^{\min}
    =
    \min_{k\in\{1,\ldots,K\}}
    R_{{\rm sec},k}.
    \label{eq:worst_user_secrecy_rate}
\end{equation}

\subsection{Alice-Side Monostatic Sensing Model of Eve}
\label{subsec:alice_sensing}

Alice uses the echoes of its downlink waveform to estimate the
directions of the passive Eves. Let $\theta_{AE,\ell}$ denote the
physical angle from Alice toward the $\ell$-th Eve. The steering
vector of Alice's ULA is
\begin{equation}
    \mathbf a_A(\theta)
    =
    \frac{1}{\sqrt{N_t}}
    \begin{bmatrix}
        1 &
        e^{j\pi\sin(\theta)} &
        \cdots &
        e^{j\pi(N_t-1)\sin(\theta)}
    \end{bmatrix}^{T}.
\end{equation}

Since Alice employs the same ULA for transmission and reception, the monostatic round-trip array-response matrix associated with angle $\theta$ is
\begin{equation}
    \mathbf A(\theta)
    =
    \mathbf a_A^*(\theta)
    \mathbf a_A^T(\theta).
\end{equation}
Thus, the received sensing signal of the $\ell$-th Eve at snapshot $q$ is modeled as
\begin{equation}
    \mathbf y_{A,\ell}[q]
    =
    \alpha_{A,\ell}
    \mathbf A(\theta_{AE,\ell})
    \mathbf x[q]
    +
    \mathbf n_{A,\ell}[q],
    \qquad
    q=1,\ldots,L_A,
    \label{eq:alice_sensing_model}
\end{equation}
where $L_A$ denotes the number of sensing snapshots, $\mathbf{x}[q]\in\mathbb{C}^{N_t}$ is the transmit vector at the $q$-th sensing snapshot, 
$\alpha_{A,\ell}\in\mathbb{C}$ is the effective complex reflection coefficient that incorporates the monostatic round-trip propagation attenuation and the effective RCS of Eve $\ell$, assumed to remain constant over the $L_A$ sensing snapshots
and 
\begin{equation} \mathbf{n}_{A,\ell}[q] \sim \mathcal{CN}\!\left( \mathbf{0}, \sigma_{rA}^{2}\mathbf{I}_{N_t} \right) \end{equation}
is the corresponding sensing-noise vector.

\subsection{Eve-Side Passive Sensing of Bobs}
\label{subsec:eve_passive_sensing}

In addition to attempting to decode the confidential downlink information, Eve $\ell$ may exploit the downlink-induced reflection from the physical user associated with Bob $k$ for passive angular sensing.
Let $\theta_{AB,k}$ denote the angle from Alice toward the $k$-th Bob and
$\theta_{EB,\ell k}$ the angle from the $\ell$-th Eve toward the
$k$-th Bob, given by
\begin{equation}
\begin{aligned}
    \theta_{AB,k}
    &=
    \operatorname{atan2}
    \left(
        y_{B,k}-y_A,
        x_{B,k}-x_A
    \right),
    \\
    \theta_{EB,\ell k}
    &=
    \operatorname{atan2}
    \left(
        y_{B,k}-y_{E,\ell},
        x_{B,k}-x_{E,\ell}
    \right),
\end{aligned}
\end{equation}
where $\operatorname{atan2}(y,x)$ denotes the four-quadrant inverse
tangent, which returns the angle of the vector $[x,y]^T$ with respect to the positive $x$-axis.
The steering vector of Eve's ULA is
\begin{equation}
    \mathbf a_E(\theta)
    =
    \frac{1}{\sqrt{N_e}}
    \begin{bmatrix}
        1 &
        e^{j\pi\sin(\theta)} &
        \cdots &
        e^{j\pi(N_e-1)\sin(\theta)}
    \end{bmatrix}^{T}.
    \label{ae}
\end{equation}

The directional transmit energy toward Bob $k$ is
\begin{equation}
    I_k(\mathbf R_x)
    =
    \mathbf a_A^H(\theta_{AB,k})
    \mathbf R_x
    \mathbf a_A(\theta_{AB,k}).
\end{equation}

The passive-sensing observation at Eve $\ell$ associated with the
$k$-th Bob is modeled as
\begin{equation}
    \mathbf y_{E,\ell k}^{\rm ps}
    =
    \mathbf a_E(\theta_{EB,\ell k})r_{\ell k}
    +
    \mathbf H_{AE,\ell}\mathbf x
    +
    \mathbf n_{pE,\ell},
\end{equation}
where $r_{\ell k}$ denotes the effective reflected waveform and
satisfies
\begin{equation}
    \mathbb E
    \left\{
        |r_{\ell k}|^2
    \right\}
    =
    \rho_{\ell k}I_k(\mathbf R_x),
    \label{rho}
\end{equation}
where $\rho_{\ell k}\geq 0$ denotes the effective bistatic reflected-power coefficient, accounting for the Alice--Bob propagation attenuation, the effective RCS of the physical user associated with Bob $k$, and the Bob--Eve propagation attenuation.
Moreover,
\begin{equation}
    \mathbf n_{pE,\ell}
    \sim
    \mathcal{CN}
    \left(
        \mathbf 0,
        \sigma_{pE}^2\mathbf I_{N_e}
    \right)
\end{equation}
is the passive-sensing noise vector.

We adopt a second-order uncorrelated-scattering model in which the
Bob-reflected component is uncorrelated with the direct Alice--Eve
component and the receiver noise. This model is consistent with an independent zero-mean random-phase scattering
coefficient on the Bob-reflected path. The covariance of the direct
Alice--Eve component is therefore
\begin{equation}
    \mathbf R_{{\rm dir},\ell}
    =
    \mathbf H_{AE,\ell}
    \mathbf R_x
    \mathbf H_{AE,\ell}^H.
\end{equation}
Accordingly, the covariance of the passive-sensing observation is
\begin{equation}
\begin{aligned}
    \mathbf R_{E,\ell k}^{\rm ps}
    &=
    \;
    \rho_{\ell k}
    I_k(\mathbf R_x)
    \mathbf a_E(\theta_{EB,\ell k})
    \mathbf a_E^H(\theta_{EB,\ell k})
    \\
    &+
    \mathbf R_{{\rm dir},\ell}
    +
    \sigma_{pE}^2\mathbf I_{N_e}.
\end{aligned}
\end{equation}
Thus, for a candidate scan angle $\vartheta$, Eve $\ell$ evaluates
the spatial response
\begin{equation}
\begin{aligned}
    P_{\ell k}(\vartheta)
    &=
    \;
    \rho_{\ell k}
    I_k(\mathbf R_x)
    \left|
        \mathbf a_E^H(\vartheta)
        \mathbf a_E(\theta_{EB,\ell k})
    \right|^2
    \\
    &+
    \mathbf a_E^H(\vartheta)
    \mathbf R_{{\rm dir},\ell}
    \mathbf a_E(\vartheta)
    +
    \sigma_{pE}^2.
\end{aligned}
\end{equation}
The dominant peak of $P_{\ell k}(\vartheta)$ determines the direction
inferred by Eve $\ell$ for the $k$-th Bob.

\section{Eve-Angle Estimation and Uncertainty Modeling}
\label{sec:BayesianUncertainty}

Alice senses the direction of each Eve through the monostatic
sensing model introduced in Section~\ref{subsec:alice_sensing}. The Eve direction is modeled using an angular prior whose support defines the uncertainty sector considered in the robust design. This section first derives the corresponding BCRB and then propagates the angular uncertainty consistently to the Alice--Eve channels, the Bob directions observed by the Eves, and the prescribed virtual
directions.

\subsection{BCRB of the Eve Direction}

Let $\theta_{AE,\ell}$ be the unknown direction of Eve $\ell$ from Alice. It is modeled as a random variable with prior probability
density $p_\ell(\theta)$ over
$\theta\in[\theta_{\ell,1},\theta_{\ell,2}]$. 
The BCRB associated with Alice's estimation of
$\theta_{AE,\ell}$ is characterized in the following lemma.

\begin{lemma}
\label{lem:eve_bcrb}
For the monostatic sensing model in
\eqref{eq:alice_sensing_model}, define the unknown parameter vector
associated with the $\ell$-th Eve as
\begin{equation}
    \boldsymbol{\eta}_{\ell}
    =
    \left[
        \theta_{AE,\ell},
        \Re\{\alpha_{A,\ell}\},
        \Im\{\alpha_{A,\ell}\}
    \right]^T.
\end{equation} 
Since $\alpha_{A,\ell}$ is independent of $\theta_{AE,\ell}$ and the transmitted signal, and is circularly symmetric, we have \begin{equation}
\mathbb{E}\{\alpha_{A,\ell}\}=0, \qquad \xi_{A,\ell} \triangleq \mathbb{E}\!\left\{|\alpha_{A,\ell}|^{2}\right\}.
\end{equation}
Then, the BCRB of the direction of the $\ell$-th Eve is
\begin{equation}
    \mathrm{BCRB}_{\ell}(\mathbf R_x)
    =
    \left[
        \operatorname{Tr}
        \left(
            \mathbf Q_{B,\ell}\mathbf R_x
        \right)
        +
        J_{P,\ell}
    \right]^{-1},
    \label{eq:bcrb_lemma}
\end{equation}
where
\begin{equation}
    J_{P,\ell}
    =
    \int_{\theta_{\ell,1}}^{\theta_{\ell,2}}
    \left(
        \frac{
            \partial\log p_\ell(\theta)
        }{
            \partial\theta
        }
    \right)^2
    p_\ell(\theta)
    \,d\theta
    \label{eq:prior_fisher_lemma}
\end{equation}
is the prior Fisher information, and
\begin{equation}
    \mathbf Q_{B,\ell}
    =
    \int_{\theta_{\ell,1}}^{\theta_{\ell,2}}
    \frac{
        2L_A\xi_{A,\ell}
    }{
        \sigma_{rA}^2
    }
    \mathbf Q_A(\theta)
    p_\ell(\theta)
    \,d\theta
    \label{eq:QB_lemma}
\end{equation}
is the prior-averaged sensing-information matrix, with
\begin{equation}
    \mathbf Q_A(\theta)
    =
    \dot{\mathbf A}^{H}(\theta)
    \dot{\mathbf A}(\theta),
    \qquad
    \dot{\mathbf A}(\theta)
    =
    \frac{
        \partial\mathbf A(\theta)
    }{
        \partial\theta
    }.
\end{equation}
\end{lemma}

\begin{IEEEproof}
The proof is presented in Appendix~\ref{app:bcrb_proof}.
\end{IEEEproof}

It follows from \eqref{eq:bcrb_lemma} that the sensing information
contributed by the transmitted waveform is linear in the covariance
matrix $\mathbf R_x$, while the total Bayesian information is affine
in $\mathbf R_x$ due to the additive prior-information term
$J_{P,\ell}$. Therefore, increasing
$\operatorname{Tr}(\mathbf Q_{B,\ell}\mathbf R_x)$ decreases the
BCRB and improves Alice's Eve-angle estimation accuracy.

\subsection{Impact of Angular Uncertainty on Geometry}

\begin{figure}
    \centering
    \includegraphics[width=0.8\columnwidth]{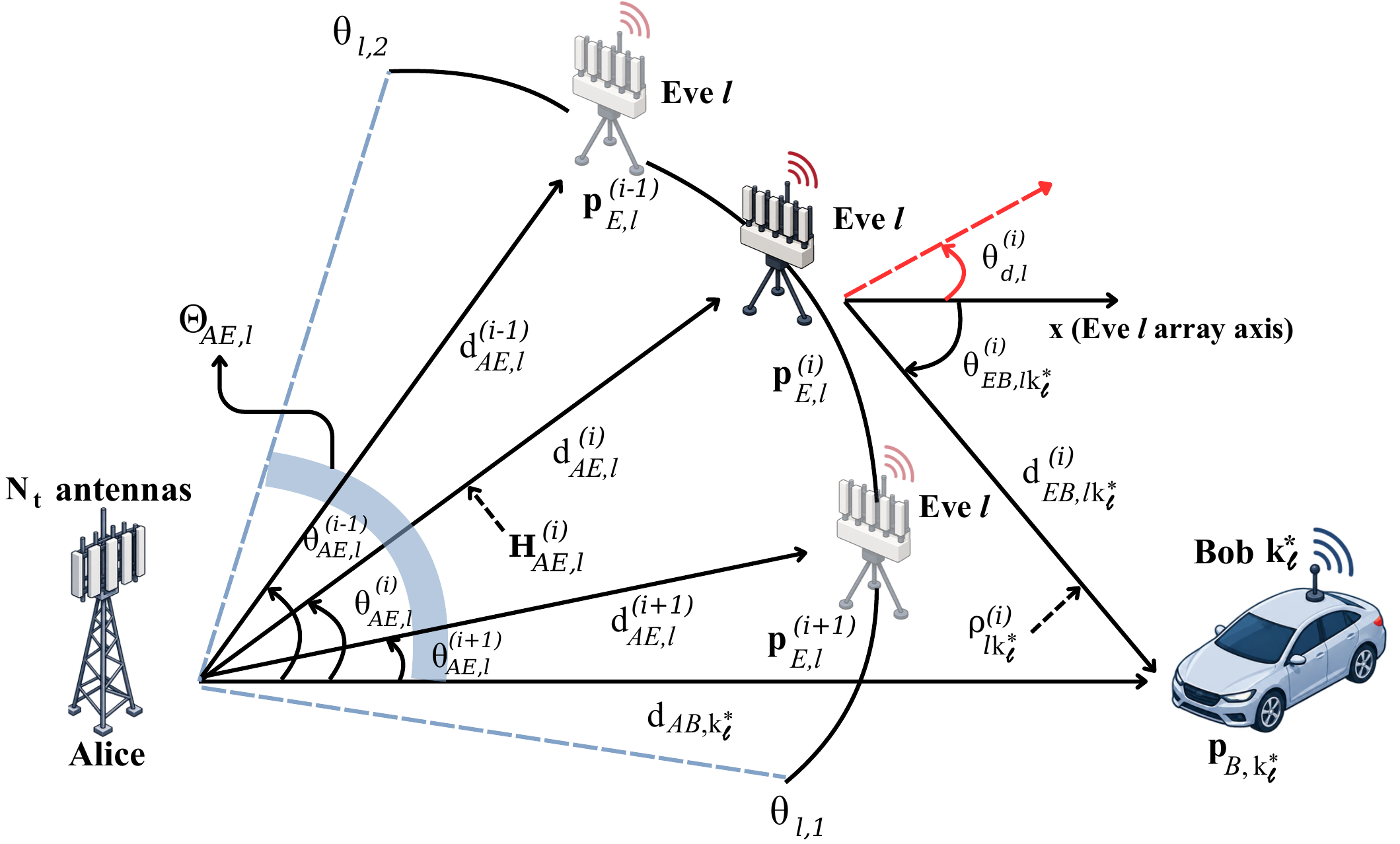}
    \vspace{-2mm}
    \caption{Geometry of Eve-angle uncertainty and deception.}
    \vspace{-4mm}
    \label{fig:uncertainty_geometry}
\end{figure}

The support $[\theta_{\ell,1},\theta_{\ell,2}]$ of the prior PDF $p_\ell(\theta)$ represents the angular directions considered possible for Eve $\ell$ and is therefore adopted as the uncertainty
sector for the robust design, as illustrated in Fig.~\ref{fig:uncertainty_geometry} \cite{Boz_iot}
\begin{equation}
    \Theta_{AE,\ell}
    \triangleq
    \left[
        \theta_{\ell,1},
        \theta_{\ell,2}
    \right].
    \label{eq:eve_uncertainty_sector}
\end{equation}
Alice is assumed to have the prior $p_\ell(\theta)$, the range $d_{AE,\ell}$, and the adopted array/channel model, while the exact Eve direction and receive combiner are unknown.
The continuous sector is discretized using $N_s\geq2$ equally spaced
angular samples:
\begin{equation}
    \theta_{AE,\ell}^{(i)}
    =
    \theta_{\ell,1}
    +
    \frac{i-1}{N_s-1}
    \left(
        \theta_{\ell,2}
        -
        \theta_{\ell,1}
    \right),
    \qquad
    i=1,\ldots,N_s.
    \label{eq:eve_angle_samples}
\end{equation}
The resulting grid includes both endpoints and is used to approximate
the robust constraints over the continuous uncertainty sector.

For each angular sample $\theta_{AE,\ell}^{(i)}$, the corresponding candidate position
of Eve $\ell$ is therefore reconstructed as
\begin{equation}
    \mathbf p_{E,\ell}^{(i)}
    =
    \mathbf p_A
    +
    d_{AE,\ell}
    \begin{bmatrix}
        \cos\theta_{AE,\ell}^{(i)}
        \\
        \sin\theta_{AE,\ell}^{(i)}
    \end{bmatrix}
    =
    \begin{bmatrix}
        x_{E,\ell}^{(i)}
        \\
        y_{E,\ell}^{(i)}
    \end{bmatrix},
    \label{eq:sampled_eve_position}
\end{equation}
where $d_{AE,\ell}$ denotes the Alice--Eve $\ell$ range available at Alice.
The Alice--Eve channel associated with the sampled position $\mathbf p_{E,\ell}^{(i)}$ is denoted by $\mathbf H_{AE,\ell}^{(i)}\in\mathbb C^{N_e\times N_t}$. Under the same geometric realization, the direction from Eve $\ell$ toward Bob $k$ is
\begin{equation}
    \theta_{EB,\ell k}^{(i)}
    =
    \operatorname{atan2}
    \left(
        y_{B,k}-y_{E,\ell}^{(i)},
        x_{B,k}-x_{E,\ell}^{(i)}
    \right).
    \label{eq:sampled_eb_angle}
\end{equation}
The corresponding Bob--Eve distance is $d_{EB,\ell k}^{(i)} =\|\mathbf p_{B,k}-\mathbf p_{E,\ell}^{(i)}\|_2$, yielding the reflected-power coefficient $\rho_{\ell k}^{(i)}$.

To select the Bob for which the virtual-direction response is
enforced, let
\begin{equation}
    \overline{\theta}_{AE,\ell}
    =
    \mathbb E_{p_\ell(\theta)}
    \left\{
        \theta
    \right\}
\end{equation}
denote the prior-mean direction of Eve $\ell$. The reference Eve
position associated with this direction is
\begin{equation}
    \overline{\mathbf p}_{E,\ell}
    =
    \mathbf p_A
    +
    d_{AE,\ell}
    \begin{bmatrix}
        \cos\overline{\theta}_{AE,\ell}
        \\
        \sin\overline{\theta}_{AE,\ell}
    \end{bmatrix}.
\end{equation}
For each Eve, the sensing-deception mechanism is applied to one selected Bob, denoted by $k_\ell^\star$. We adopt the Bob closest to the reference Eve position as a geometry-based selection rule, i.e., 
\begin{equation}
    k_\ell^\star
    =
    \operatorname*{arg\,min}_{k\in\{1,\ldots,K\}}
    \left\|
        \mathbf p_{B,k}
        -
        \overline{\mathbf p}_{E,\ell}
    \right\|_2.
    \label{eq:threat_bob_selection}
\end{equation}
The distance-based criterion is used for user selection only and does not imply that the selected Bob necessarily produces the strongest instantaneous reflected signal. The selection is performed independently for each Eve and does not impose a one-to-one correspondence between Eves and Bobs. Thus, multiple Eves may select the same Bob.

For each angular sample, the virtual direction prescribed for Eve
$\ell$ is
\begin{equation}
    \theta_{d,\ell}^{(i)}
    =
    \operatorname{wrap}_{\pi}
    \left(
        \theta_{EB,\ell k_\ell^\star}^{(i)}
        +
        \Delta_g
    \right),
    \label{eq:sampled_ghost_angle}
\end{equation}
where $\Delta_g$ is a fixed angular offset selected to provide a meaningful separation from the corresponding Bob direction while ensuring that the resulting virtual directions remain within the considered scan region of Eve's ULA. The prescribed deceptive direction does not correspond to the bearing of a physical reflector or scatterer. Rather, it denotes the scan angle at which the covariance design forces Eve’s passive angular response $P_{\ell k}(\vartheta)$ to exhibit a dominant artificial peak. 

Consequently, for each angular sample $\theta_{AE,\ell}^{(i)}$, the potential Eve position $\mathbf p_{E,\ell}^{(i)}$ jointly determines the sampled Alice--Eve channel $\mathbf H_{AE,\ell}^{(i)}$, the Bob--Eve direction $\theta_{EB,\ell k}^{(i)}$, the reflected-power coefficient $\rho_{\ell k}^{(i)}$, and the virtual direction $\theta_{d,\ell}^{(i)}$. Hence, all these quantities reflect the same underlying angular uncertainty and geometric realization.

\section{Problem Formulation and Proposed Solution}
\label{sec:Optsol}

This section presents the optimization design used to jointly secure the downlink communication and mislead Eve's passive sensing process. The objective balances four goals: maximizing the worst-user secrecy margin, strengthening the prescribed ghost response relative to competing scan directions, improving Alice's Eve-angle estimation through the BCRB metric, and limiting the power assigned to the deception signals. To obtain a tractable formulation, the continuous Eve-decoding and passive-scan requirements are reformulated, the nonconcave secrecy-margin term is handled through SCA, and the rank-one constraints are relaxed, yielding a convex semidefinite-constrained subproblem at each iteration.

\subsection{Problem Formulation}
\label{subsec:original_problem}

The SRs in \eqref{eq:bob_secrecy_rate}--\eqref{eq:worst_user_secrecy_rate} depend on the decoding capabilities of all Eves. The robust Eve data-decoding constraint introduced below limits the maximum post-combining data-decoding SINR of every Eve for every information stream to $\Gamma_E$. Consequently,
\begin{equation}
    R_{E,\ell k}
    \leq
    \log_2(1+\Gamma_E),
    \qquad
    \forall\ell,k.
\end{equation}
We therefore define the corresponding worst-user secrecy margin as
\begin{equation}
    \overline{R}_{\rm sec}^{\rm wc}
    =
    \min_{k\in\{1,\ldots,K\}}
    \left\{
        R_{B,k}
        -
        \log_2(1+\Gamma_E)
    \right\}.
    \label{eq:worst_user_secrecy_margin}
\end{equation}
Since the positive-part operator is monotone, the associated conservative lower bound on the physical worst-user SR is
\begin{equation}
    R_{\rm sec}^{\rm wc}
    =
    \left[
        \overline{R}_{\rm sec}^{\rm wc}
    \right]^+.
    \label{eq:original_worst_secrecy}
\end{equation}
To avoid introducing the nonsmooth positive-part operator into the
optimization, the covariance design maximizes the untruncated margin
$\overline{R}_{\rm sec}^{\rm wc}$. When this margin is nonnegative, it
coincides with the conservative physical SR lower bound, otherwise, the latter is zero.

The BCRB derived in
Section~\ref{sec:BayesianUncertainty} lower-bounds the mean-squared
error of Alice's Eve-angle estimator. To combine the sensing term
with the other objectives on a dimensionless scale, we define the
normalized BCRB term as
\begin{equation}
    \mathcal B(\mathbf R_x)
    =
    \frac{1}{L}
    \sum_{\ell=1}^{L}
    \frac{
        \mathrm{BCRB}_{\ell}(\mathbf R_x)
    }{
        \mathrm{BCRB}_{\ell}^{\rm ref}
    },
    \label{eq:normalized_bcrb}
\end{equation}
where
\begin{equation}
    \mathrm{BCRB}_{\ell}^{\rm ref}
    =
    \mathrm{BCRB}_{\ell}
    \left(
        \frac{P_{\max}}{N_t}\mathbf I_{N_t}
    \right)
    \label{eq:bcrb_reference}
\end{equation}
is evaluated under the isotropic full-power covariance
$(P_{\max}/N_t)\mathbf I_{N_t}$ and is used solely for
normalization.

The secrecy-margin and deceptive-peak separation terms may also have
substantially different numerical scales. The positive constants
$R_{\rm sec}^{\rm ref}$ and $\Delta^{\rm ref}$ denote fixed
normalization values for the worst-user secrecy margin and the
deceptive-peak separation, respectively. Their computation is
specified after the tractable reformulations are developed. These
values are computed once and kept fixed throughout the joint
optimization. Similarly,
$\operatorname{Tr}(\mathbf Z)/P_{\max}$ represents the
fraction of the transmit-power budget allocated to the deception
signals. 

To represent the Alice--Eve channel variation over the angular uncertainty sector, let $\mathbf H_{AE,\ell}(\theta)$ denote the channel corresponding to the candidate Eve direction $\theta\in\Theta_{AE,\ell}$ at the assumed Alice--Eve range $d_{AE,\ell}$. Accordingly, the sampled channel matrices satisfy
\begin{equation}
    \mathbf H_{AE,\ell}^{(i)}
    =
    \mathbf H_{AE,\ell}
    \left(
        \theta_{AE,\ell}^{(i)}
    \right).
\end{equation}

Using the passive angular response defined in Section~\ref{subsec:eve_passive_sensing}, the design enforces a dominant response at the prescribed deceptive direction for the selected Bob $k_\ell^\star$. Let $\Theta_E$ denote the visible angular region of the Eve ULA, and let $\epsilon_g>0$ specify the angular neighborhood regarded as part of the same deceptive peak. For the $\ell$-th Eve and geometry sample $i$, define the competing scan-angle region as
\begin{equation}
    \Theta_{\ell}^{g,(i)}
    =
    \Theta_E
    \setminus
    \left(
        \theta_{d,\ell}^{(i)}-\epsilon_g,\,
        \theta_{d,\ell}^{(i)}+\epsilon_g
    \right),
    \label{eq:competing_scan_region}
\end{equation}
where $\setminus$ denotes the set-difference operator.

Therefore, the original robust multiobjective covariance design is formulated as
\begin{equation}
\tag{\textbf{P0}}
\label{eq:original_opt}
\begin{array}{cl}
\displaystyle
\mathop{\max}_{\substack{
\{\mathbf W_k\}_{k=1}^{K},\;
\mathbf Z,\;
\Delta
}}
&
\begin{aligned}
&
\lambda_s
\frac{\overline{R}_{\rm sec}^{\rm wc}}{R_{\rm sec}^{\rm ref}}
+
\lambda_g
\frac{\Delta}{\Delta^{\rm ref}}
\\[-1mm]
&
-
\lambda_b\mathcal B(\mathbf R_x)
-
\lambda_z
\frac{\operatorname{Tr}(\mathbf Z)}{P_{\max}}
\end{aligned}
\\[2ex]
\mathrm{s.t.}
&
\mathrm{C}_1:\;
\operatorname{Tr}(\mathbf R_x)
\leq
P_{\max},
\\[1ex]
&
\mathrm{C}_2:\;
S_{B,k}
\geq
\gamma_{B,\min}I_{B,k},
\qquad
\forall k,
\\[1ex]
&
\mathrm{C}_3:\;
\begin{aligned}[t]
&
\mathbf H_{AE,\ell}(\theta)
\mathbf W_k
\mathbf H_{AE,\ell}^{H}(\theta)
\\[-0.5ex]
&
\preceq
\Gamma_E
\Bigl(
    \mathbf H_{AE,\ell}(\theta)
    \mathbf X_k
    \mathbf H_{AE,\ell}^{H}(\theta)
\\[-0.5ex]
&
\,
+
\sigma_e^2\mathbf I_{N_e}
\Bigr),
\quad
\forall\theta\in\Theta_{AE,\ell},
\,
\forall\ell,k,
\end{aligned}
\\[1ex]
&
\mathrm{C}_4:\;
P_{\ell k_\ell^\star}^{(i)}
\left(
    \theta_{d,\ell}^{(i)}
\right)
\geq
P_{\ell k_\ell^\star}^{(i)}(\vartheta)
+
\Delta,
\\
&\qquad
\forall\vartheta\in\Theta_{\ell}^{g,(i)},
\quad
\forall\ell,i,
\\[1ex]
&
\mathrm{C}_5:\;
\mathbf W_k\succeq\mathbf 0,
\quad
\operatorname{rank}(\mathbf W_k)=1,
\quad
\forall k,
\\[1ex]
&
\mathrm C_6:\;
\mathbf Z\succeq\mathbf 0,\qquad
\mathrm{C}_7:\;
\Delta\geq0.
\end{array}
\end{equation}

The nonnegative weights $ \lambda_s, \lambda_g, \lambda_b, \lambda_z$ control the relative emphasis on secrecy, deceptive-peak separation, Alice’s Eve-direction estimation accuracy, and deception-power consumption, respectively, and satisfy
\begin{equation}
    \lambda_s+\lambda_g+\lambda_b+\lambda_z=1.
    \label{eq:objective_weights}
\end{equation}
Constraint $\mathrm{C}_1$ imposes the total transmit-power budget, while $\mathrm{C}_2$ guarantees the prescribed SINR requirement at every Bob. Constraint $\mathrm{C}_3$ limits the data-decoding capability of every Eve throughout its continuous angular uncertainty sector. Unlike $\mathrm{C}_3$, $\mathrm{C}_4$ is enforced at the sampled Eve geometries. For each such geometry, $\mathrm{C}_4$ requires the response at the prescribed deceptive direction to exceed every competing passive-scan response by at least $\Delta$. Finally, constraints $\mathrm{C}_5$--$\mathrm{C}_7$ impose the beamforming structure, positive semidefiniteness of the deception covariances, and nonnegativity of the deceptive-peak separation, respectively.

Problem~\eqref{eq:original_opt} is nonconvex and semi-infinite. The secrecy-margin term contains nonconcave rate differences, while constraint $\mathrm{C}_3$ must hold for every Alice--Eve direction in a continuous uncertainty sector. Constraint $\mathrm{C}_4$ contains infinitely many inequalities because the competing scan angle $\vartheta$ varies continuously, and the rank-one requirements in $\mathrm{C}_5$ are nonconvex. Hence, the following subsection derives tractable reformulations of these components.

\subsection{Tractable Reformulations}
\label{subsec:tractable_reformulations}

\subsubsection{Secrecy Rate SCA Reformulation}
\label{subsec:sca_rate}

We first derive a concave lower bound for the nonconcave worst-user
secrecy margin in the objective of
problem~\eqref{eq:original_opt}. Using the signal and interference
powers defined in the communication model, the achievable rate of
Bob $k$ can be written as
$R_{B,k}
=
\log_2(S_{B,k}+I_{B,k})
-
\log_2(I_{B,k})$.
Although $S_{B,k}$ and $I_{B,k}$ are affine in the covariance
matrices, this difference of concave logarithmic functions is not
concave in general.

At SCA iteration $r$, let $I_{B,k}^{(r)}$ denote the value of
$I_{B,k}$ evaluated at the current covariance point. Since
$\log_2(I_{B,k})$ is concave, its first-order Taylor expansion at
$I_{B,k}^{(r)}$ is a global upper bound
\begin{equation}
    \log_2(I_{B,k})
    \leq
    \log_2(I_{B,k}^{(r)})
    +
    \frac{
        I_{B,k}-I_{B,k}^{(r)}
    }{
        I_{B,k}^{(r)}\ln 2
    }.
\end{equation}
Substituting this upper bound into the difference expression of
$R_{B,k}$ yields the concave lower bound
\begin{equation}
\begin{aligned}
    R_{B,k}^{\rm lb,(r)}
    &=
    \;
    \log_2(S_{B,k}+I_{B,k})
    -
    \log_2(I_{B,k}^{(r)})
    -
    \frac{
        I_{B,k}-I_{B,k}^{(r)}
    }{
        I_{B,k}^{(r)}\ln 2
    }.
    \label{eq:rate_lower_bound}
\end{aligned}
\end{equation}
To represent the worst-user secrecy margin, we introduce the auxiliary
variable $t_{\rm sec}$ and impose
\begin{equation}
    t_{\rm sec}
    \leq
    R_{B,k}^{\rm lb,(r)}
    -
    \log_2(1+\Gamma_E),
    \qquad
    \forall k.
    \label{eq:sca_secrecy_constraint}
\end{equation}
Since $R_{B,k}^{\rm lb,(r)}$ is concave for a fixed linearization
point, \eqref{eq:sca_secrecy_constraint} defines the hypograph of a
concave function and is therefore convex. Maximizing $t_{\rm sec}$
therefore maximizes a conservative SCA lower bound on the worst-user
secrecy margin. The corresponding physical SR lower bound
is subsequently obtained by applying the positive-part operator to
the achieved worst-user margin, as in
\eqref{eq:original_worst_secrecy}.

\subsubsection{Robust Eve Data-Decoding Reformulation}
\label{subsubsec:robust_eve_reformulation}

We now reformulate the semi-infinite Eve-decoding constraint
$\mathrm{C}_3$ of problem~\eqref{eq:original_opt}. For Eve $\ell$ attempting to decode the information stream intended for Bob $k$, the useful-signal covariance under angular sample $i$ is
\begin{equation}
    \mathbf R_{E,\ell k}^{\rm sig,(i)}
    =
    \mathbf H_{AE,\ell}^{(i)}
    \mathbf W_k
    \mathbf H_{AE,\ell}^{(i)H},
\end{equation}
whereas the corresponding interference-plus-noise covariance is
\begin{equation}
    \mathbf R_{E,\ell k}^{\rm int,(i)}
    =
    \mathbf H_{AE,\ell}^{(i)}
    \mathbf X_k
    \mathbf H_{AE,\ell}^{(i)H}
    +
    \sigma_e^2\mathbf I_{N_e}.
\end{equation}
Since
$\mathbf R_{E,\ell k}^{\rm int,(i)}\succ\mathbf 0$, the maximum
post-combining SINR available to Eve $\ell$ satisfies
\begin{equation}
\begin{aligned}
    &
    \lambda_{\max}
    \left(
        \left(
            \mathbf R_{E,\ell k}^{\rm int,(i)}
        \right)^{-1/2}
        \mathbf R_{E,\ell k}^{\rm sig,(i)}
        \left(
            \mathbf R_{E,\ell k}^{\rm int,(i)}
        \right)^{-1/2}
    \right)
    \leq
    \Gamma_E
    \\
    &\Longleftrightarrow\quad
    \mathbf R_{E,\ell k}^{\rm sig,(i)}
    \preceq
    \Gamma_E
    \mathbf R_{E,\ell k}^{\rm int,(i)}.
\end{aligned}
\end{equation}
For fixed sampled channels, this condition is an LMI in the information and deception covariance matrices.

To account for channel variations between adjacent angular samples, let
\begin{equation}
    h_\ell
    =
    \frac{
        \theta_{\ell,2}-\theta_{\ell,1}
    }{
        N_s-1
    }
\end{equation}
denote the angular grid spacing. The geometry-dependent Alice--Eve channel is assumed to be differentiable over $\Theta_{AE,\ell}$, and let $L_{H,\ell}$ satisfy
\begin{equation}
    L_{H,\ell}
    \geq
    \sup_{\theta\in\Theta_{AE,\ell}}
    \left\|
        \frac{
            \partial\mathbf H_{AE,\ell}(\theta)
        }{
            \partial\theta
        }
    \right\|_2.
    \label{eq:channel_angular_lipschitz}
\end{equation}
Since every angle in the uncertainty sector lies within $h_\ell/2$ of a grid point, the corresponding channel variation satisfies
\begin{equation}
    \left\|
        \mathbf H_{AE,\ell}(\theta)
        -
        \mathbf H_{AE,\ell}^{(i)}
    \right\|_2
    \leq
    \varepsilon_\ell,
    \qquad
    \varepsilon_\ell
    \triangleq
    \frac{L_{H,\ell}h_\ell}{2},
\end{equation}
for the nearest angular sample $i$. Here, $\varepsilon_\ell$ is an
upper bound on the spectral-norm variation of the Alice--Eve channel
between an arbitrary angle and its nearest grid point.

For compactness, we further define
\begin{equation}
    \mathbf A_k
    =
    \mathbf W_k
    -
    \Gamma_E\mathbf X_k.
\end{equation}
For any channel $\mathbf H_{AE,\ell}(\theta) = \mathbf H_{AE,\ell}^{(i)} + \Delta\mathbf H_{\ell i}$ satisfying $|\Delta\mathbf H_{\ell i}|2\leq\varepsilon\ell$, the quadratic-term variation is bounded as
\begin{equation}
\begin{aligned}
    &
    \Big\|
        \mathbf H_{AE,\ell}(\theta)
        \mathbf A_k
        \mathbf H_{AE,\ell}^{H}(\theta)
        -
        \mathbf H_{AE,\ell}^{(i)}
        \mathbf A_k
        \mathbf H_{AE,\ell}^{(i)H}
    \Big\|_2
    \\
    &\leq
    \left(
        2\varepsilon_\ell
        \left\|
            \mathbf H_{AE,\ell}^{(i)}
        \right\|_2
        +
        \varepsilon_\ell^2
    \right)
    \left\|
        \mathbf A_k
    \right\|_2.
\end{aligned}
\label{eq:eve_quadratic_variation}
\end{equation}
Moreover, since
\begin{equation}
    \operatorname{Tr}(\mathbf W_k)
    +
    \operatorname{Tr}(\mathbf X_k)
    =
    \operatorname{Tr}(\mathbf R_x)
    \leq
    P_{\max},
\end{equation}
we have
\begin{equation}
\begin{aligned}
    \left\|\mathbf A_k\right\|_2
    &\leq
    \left\|\mathbf W_k\right\|_2
    +
    \Gamma_E
    \left\|\mathbf X_k\right\|_2
    \leq
    \max\{1,\Gamma_E\}P_{\max}.
\end{aligned}
\end{equation}
Combining \eqref{eq:eve_quadratic_variation} with the power
bound on $\|\mathbf A_k\|_2$, the variation of the quadratic matrix
term is upper-bounded by the fixed scalar
\begin{equation}
    \delta_{\ell i}^{E}
    =
    \left(
        2\varepsilon_\ell
        \left\|
            \mathbf H_{AE,\ell}^{(i)}
        \right\|_2
        +
        \varepsilon_\ell^2
    \right)
    \max\{1,\Gamma_E\}P_{\max}.
    \label{eq:eve_intersample_margin}
\end{equation}
A sufficient condition for limiting Eve's decoding SINR by
$\Gamma_E$ throughout the continuous uncertainty sector is therefore
\begin{equation}
\begin{aligned}
    &
    \mathbf H_{AE,\ell}^{(i)}
    \mathbf W_k
    \mathbf H_{AE,\ell}^{(i)H}
    +
    \delta_{\ell i}^{E}\mathbf I_{N_e}
    \\
    &\preceq
    \Gamma_E
    \left(
        \mathbf H_{AE,\ell}^{(i)}
        \mathbf X_k
        \mathbf H_{AE,\ell}^{(i)H}
        +
        \sigma_e^2\mathbf I_{N_e}
    \right),
    \qquad
    \forall \ell,k,i.
\end{aligned}
\label{eq:robust_eve_decoding_lmi}
\end{equation}
The margin $\delta_{\ell i}^{E}$ accounts for the maximum channel
variation between angular samples. Since it depends only on the
sampled channels, the uncertainty-sector geometry, and fixed system
parameters, \eqref{eq:robust_eve_decoding_lmi} remains an LMI in
$\{\mathbf W_k\}_{k=1}^{K}$ and
$\mathbf Z$ and provides a finite sufficient
replacement for constraint $\mathrm{C}_3$.

\subsubsection{Continuous Ghost-Dominance Reformulation}
\label{subsec:ghost_reformulation}

We next derive an exact finite-dimensional reformulation of the
continuous scan-angle requirement in constraint $\mathrm C_4$ for each sampled
Eve geometry. Depending on the location of
$\theta_{d,\ell}^{(i)}$ within $\Theta_E$, the competing scan-angle
region defined in \eqref{eq:competing_scan_region} consists of one or
two nonempty closed intervals and can be expressed as
\begin{equation}
    \Theta_{\ell}^{g,(i)}
    =
    \bigcup_{s\in\mathcal S_{\ell i}}
    \Theta_{\ell i,s}^{g},
    \qquad
    \Theta_{\ell i,s}^{g}
    =
    \left[
        \vartheta_{\ell i,s}^{-},
        \vartheta_{\ell i,s}^{+}
    \right],
    \label{eq:competing_scan_intervals}
\end{equation}
where $\bigcup$ denotes the union of indexed intervals, $\mathcal S_{\ell i}\subseteq{-,+}$ indexes nonempty intervals, and $\vartheta_{\ell i,s}^{-}$ and $\vartheta_{\ell i,s}^{+}$ are the lower and upper endpoints of interval $s$, respectively.

Constraint $\mathrm C_4$ requires the passive response at the
prescribed deceptive direction $\theta_{d,\ell}^{(i)}$ to exceed
every competing response over $\Theta_{\ell}^{g,(i)}$ by at least
$\Delta$, for each Eve $\ell$ and geometry sample $i$. For any fixed
scan angle $\vartheta$, this requirement is affine in the information
and deception covariance matrices. However, it is semi-infinite
because $\vartheta$ varies continuously over
$\Theta_{\ell}^{g,(i)}$.

For Eve's half-wavelength ULA, define the spatial frequency
$u=\pi\sin\vartheta$ and, using the steering vector in \eqref{ae},
the corresponding representation as
$\widetilde{\mathbf a}_E(u)\triangleq
\mathbf a_E(\arcsin(u/\pi))$.

Using the passive-response model in
Section~\ref{subsec:eve_passive_sensing}, the two angle-dependent
quadratic terms associated with Bob $k_\ell^\star$ can be collected
into the effective covariance
\begin{equation}
\begin{aligned}
    \mathbf C_{\ell}^{(i)}
    &=
    \;
    \rho_{\ell k_\ell^\star}^{(i)}
    I_{k_\ell^\star}(\mathbf R_x)
    \mathbf a_E
    \left(
        \theta_{EB,\ell k_\ell^\star}^{(i)}
    \right)
    \mathbf a_E^H
    \left(
        \theta_{EB,\ell k_\ell^\star}^{(i)}
    \right)
    \\
    &+
    \mathbf H_{AE,\ell}^{(i)}
    \mathbf R_x
    \mathbf H_{AE,\ell}^{(i)H}.
\end{aligned}
\label{eq:effective_covariance}
\end{equation}
Ignoring the angle-independent noise, the passive response is $\widetilde{\mathbf a}_E^H(u)\mathbf C_{\ell}^{(i)} \widetilde{\mathbf a}_E(u)$.

Let now
$u_{d,\ell}^{(i)}
=
\pi\sin(\theta_{d,\ell}^{(i)})$
denote the spatial frequency associated with the prescribed deceptive
direction. The response-difference function is defined as
\begin{equation}
\begin{aligned}
    q_{\ell i}(u)
    &=
    \;
    \widetilde{\mathbf a}_E^H
    \left(
        u_{d,\ell}^{(i)}
    \right)
    \mathbf C_{\ell}^{(i)}
    \widetilde{\mathbf a}_E
    \left(
        u_{d,\ell}^{(i)}
    \right)
    -
    \widetilde{\mathbf a}_E^H(u)
    \mathbf C_{\ell}^{(i)}
    \widetilde{\mathbf a}_E(u)
    -
    \Delta.
\end{aligned}
\label{eq:response_difference_trig_poly}
\end{equation}
The passive-noise term cancels because it is identical at the
prescribed and competing scan angles.

The corresponding spatial-frequency interval is
\begin{equation}
    \mathcal U_{\ell i,s}
    =
    \left[
        u_{\ell i,s}^{-},
        u_{\ell i,s}^{+}
    \right],
    \qquad
    u_{\ell i,s}^{\pm}
    =
    \pi\sin
    \left(
        \vartheta_{\ell i,s}^{\pm}
    \right).
    \label{eq:spatial_frequency_intervals}
\end{equation}
Since $\Theta_E\subseteq[-\pi/2,\pi/2]$, the transformation is
monotonic. Hence, constraint $\mathrm C_4$ is equivalent to
\begin{equation}
    q_{\ell i}(u)\geq0,
    \qquad
    \forall u\in\mathcal U_{\ell i,s},
    \quad
    \forall s\in\mathcal S_{\ell i},
    \quad
    \forall\ell,i.
    \label{eq:trig_polynomial_nonnegativity}
\end{equation}

For each spatial-frequency interval, we define
\begin{equation}
\begin{gathered}
    \bar{u}_{\ell i,s}
    =
    \frac{
        u_{\ell i,s}^{-}
        +
        u_{\ell i,s}^{+}
    }{2},
    \qquad
    \tau
    =
    \tan
    \left(
        \frac{
            u-\bar{u}_{\ell i,s}
        }{2}
    \right),
    \\[1mm]
    \tau_{\ell i,s}^{\pm}
    =
    \tan
    \left(
        \frac{
            u_{\ell i,s}^{\pm}
            -
            \bar{u}_{\ell i,s}
        }{2}
    \right).
\end{gathered}
\label{eq:tangent_half_angle}
\end{equation}
Because each spatial-frequency interval has width strictly smaller than $2\pi$, the shifted tangent half-angle transformation is finite and bijective over that interval.

We further define
\begin{equation}
    p_{\ell i,s}(\tau)
    =
    (1+\tau^2)^{N_e-1}
    q_{\ell i}
    \left(
        \bar{u}_{\ell i,s}
        +
        2\arctan\tau
    \right).
    \label{eq:ordinary_ghost_polynomial}
\end{equation}
Then, $p_{\ell i,s}(\tau)$ is a real polynomial of degree at most
$2(N_e-1)$ whose coefficients are affine in the covariance variables
and $\Delta$.

\begin{proposition}
\label{prop:continuous_ghost_reformulation}
For each Eve $\ell$, geometry sample $i$, and interval
$s\in\mathcal S_{\ell i}$, the nonnegativity condition in
\eqref{eq:trig_polynomial_nonnegativity} holds if and only if there
exist real symmetric positive-semidefinite matrices
$\mathbf Q_{\ell i,s}^{(0)}\in\mathbb S_{+}^{N_e}$ and
$\mathbf Q_{\ell i,s}^{(1)}\in\mathbb S_{+}^{N_e-1}$ such that
\begin{equation}
\begin{aligned}
    p_{\ell i,s}(\tau)
    &=
    \;
    \mathbf v_{N_e-1}^{T}(\tau)
    \mathbf Q_{\ell i,s}^{(0)}
    \mathbf v_{N_e-1}(\tau)
    \\
    &+
    \left(
        \tau-\tau_{\ell i,s}^{-}
    \right)
    \left(
        \tau_{\ell i,s}^{+}-\tau
    \right)
    \mathbf v_{N_e-2}^{T}(\tau)
    \mathbf Q_{\ell i,s}^{(1)}
    \mathbf v_{N_e-2}(\tau),
\end{aligned}
\label{eq:interval_sos_identity}
\end{equation}
where
$\mathbf v_d(\tau)
=
[1,\tau,\ldots,\tau^d]^T$.
Equality in \eqref{eq:interval_sos_identity} denotes a polynomial
identity in $\tau$. Specifically, after expanding both sides, the
coefficients of each monomial $\tau^m$, for
$m=0,\ldots,2(N_e-1)$, are equated.
\end{proposition}

\begin{IEEEproof}
The proof is provided in
Appendix~\ref{app:continuous_ghost_reformulation}.
\end{IEEEproof}

Consequently, coefficient matching in
\eqref{eq:interval_sos_identity}, together with the positive-semidefinite constraints on the Gram matrices, replaces $\mathrm C_4$ with finitely many affine equalities and LMIs for each sampled Eve geometry.

\subsubsection{BCRB Convexity and Semidefinite Relaxation}
\label{subsubsec:bcrb_sdr}

From \eqref{eq:bcrb_lemma}, each
$\mathrm{BCRB}_{\ell}(\mathbf R_x)$ is the reciprocal of a positive
affine function of $\mathbf R_x$ and is therefore convex over its
domain. Since the normalization constants in
\eqref{eq:normalized_bcrb} are strictly positive,
$\mathcal B(\mathbf R_x)$ is also convex. Consequently,
$-\lambda_b\mathcal B(\mathbf R_x)$ is concave.

After the preceding reformulations, the rank-one constraints
$\mathrm C_5$ remain nonconvex. Dropping these constraints yields a
semidefinite relaxation of the covariance-design problem.

\subsection{Final Reformulated Convex SCA Subproblem}
\label{subsec:proposed_sca_subproblem}

For notational compactness, collect the optimization variables as
\begin{equation*}
\mathcal V
\triangleq
\left\{
    \{\mathbf W_k\}_{k=1}^{K},
    \mathbf Z,
    t_{\rm sec},
    \Delta,
    \{\mathbf Q_{\ell i,s}^{(0)},
      \mathbf Q_{\ell i,s}^{(1)}\}_{\ell,i,s}
\right\}.
\end{equation*}
where the Gram matrices are included for every geometry sample $i$
and every $s\in\mathcal S_{\ell i}$.

At SCA iteration $r$, the preceding reformulations and semidefinite
relaxation yield the following convex semidefinite-constrained
subproblem
\begingroup
\small
\setlength{\arraycolsep}{2pt}
\renewcommand{\arraystretch}{1.05}

\begin{align}
\underset{\mathcal V}{\operatorname{max}}
\quad
&
\begin{aligned}[t]
&
\lambda_s
\frac{t_{\rm sec}}{R_{\rm sec}^{\rm ref}}
+
\lambda_g
\frac{\Delta}{\Delta^{\rm ref}}
-
\lambda_b
\mathcal B(\mathbf R_x)
-
\lambda_z
\frac{
    \operatorname{Tr}(\mathbf Z)
}{
    P_{\max}
}
\\
&
-
\mu
\mathcal P^{(r)}
\bigl(
    \{\mathbf W_k\},
    \mathbf Z
\bigr)
\end{aligned}
\notag
\\[1ex]
\mathrm{s.t.}
\quad
&
\widetilde{\mathrm C}_1:\;
\operatorname{Tr}(\mathbf R_x)
\leq
P_{\max},
\notag
\\
&
\widetilde{\mathrm C}_2:\;
S_{B,k}
\geq
\gamma_{B,\min}I_{B,k},
\quad
\forall k,
\notag
\\
&
\widetilde{\mathrm C}_3:\;
t_{\rm sec}
\leq
R_{B,k}^{\rm lb,(r)}
-
\log_2(1+\Gamma_E),
\quad
\forall k,
\notag
\\
&
\widetilde{\mathrm C}_4:\;
\begin{aligned}[t]
&
\mathbf H_{AE,\ell}^{(i)}
\mathbf W_k
\mathbf H_{AE,\ell}^{(i)H}
+
\delta_{\ell i}^{E}\mathbf I_{N_e}
\\[-0.5ex]
&
\preceq
\Gamma_E
\Bigl(
    \mathbf H_{AE,\ell}^{(i)}
    \mathbf X_k
    \mathbf H_{AE,\ell}^{(i)H}
+
\sigma_e^2\mathbf I_{N_e}
\Bigr),
\quad
\forall\ell,k,i,
\end{aligned}
\notag
\\
&
\widetilde{\mathrm C}_5:\;
\begin{aligned}[t]
p_{\ell i,s}(\tau)
&=
\mathbf v_{N_e-1}^{T}(\tau)
\mathbf Q_{\ell i,s}^{(0)}
\mathbf v_{N_e-1}(\tau)
\\[-0.5ex]
&+
\left(
    \tau-\tau_{\ell i,s}^{-}
\right)
\left(
    \tau_{\ell i,s}^{+}-\tau
\right)
\\[-0.5ex]
&\times
\mathbf v_{N_e-2}^{T}(\tau)
\mathbf Q_{\ell i,s}^{(1)}
\mathbf v_{N_e-2}(\tau),
\\[-0.5ex]
&
\forall\ell,i,\;
s\in\mathcal S_{\ell i},
\end{aligned}
\notag
\\
&
\widetilde{\mathrm C}_6:\;
\begin{aligned}[t]
&
\mathbf Q_{\ell i,s}^{(0)}
\succeq
\mathbf 0,
\quad
\mathbf Q_{\ell i,s}^{(1)}
\succeq
\mathbf 0,
\quad
\forall\ell,i,\;
s\in\mathcal S_{\ell i},
\end{aligned}
\notag
\\
&
\widetilde{\mathrm C}_7:\;
\mathbf W_k
\succeq
\mathbf 0,
\quad
\forall k, \qquad
\notag
\widetilde{\mathrm C}_8:\;
\mathbf Z\succeq\mathbf 0.
\notag
\\
&
\widetilde{\mathrm C}_9:\;
\Delta
\geq
0.
\tag{\ensuremath{\mathbf{P1}^{(r)}}}
\label{eq:final_opt}
\end{align}

\endgroup

In the objective of problem~\eqref{eq:final_opt},
$\mathcal P^{(r)}$ is a proximal regularization function that
discourages large deviations from the current covariance point. It
is defined as
\begin{equation}
    \mathcal P^{(r)}
    \bigl(
        \{\mathbf W_k\},
        \mathbf Z
    \bigr)
    \triangleq
    \sum_{k=1}^{K}
    \left\|
        \mathbf W_k-\mathbf W_k^{(r)}
    \right\|_F^2
    +
    \left\|
        \mathbf Z-\mathbf Z^{(r)}
    \right\|_F^2.
\end{equation}
Here, $\|\cdot\|_F$ denotes the Frobenius norm, and $\mu>0$ is the
weight of the proximal penalty in the objective of
problem~\eqref{eq:final_opt}. Since
$\mathcal P^{(r)}$ and its gradient vanish at the current SCA point,
the proximal term preserves the local value and first-order
consistency of the surrogate while improving the numerical stability
of the iterations.

The fixed normalization scales are computed once from two auxiliary
single-objective instances of the same relaxed SCA procedure. Let
$t_{\rm sec}^{\star,s}$ and $\Delta^{\star,g}$ denote the converged
values obtained by optimizing only the worst-user secrecy margin and
the deceptive-peak separation, respectively. The strictly positive
scales are selected as
\begin{equation}
    R_{\rm sec}^{\rm ref}
    =
    \max
    \left\{
        \left|t_{\rm sec}^{\star,s}\right|,
        \epsilon_{\rm ref}
    \right\},
    \qquad
    \Delta^{\rm ref}
    =
    \max
    \left\{
        \Delta^{\star,g},
        \epsilon_{\rm ref}
    \right\},
    \label{eq:objective_reference_values}
\end{equation}
where $\epsilon_{\rm ref}>0$ is a small fixed scaling floor. These
reference values are then kept fixed in the joint multiobjective
design.

For a given SCA point
$\{\mathbf W_k^{(r)}\}_{k=1}^{K}$ and $\mathbf Z^{(r)}$,
problem~\eqref{eq:final_opt} is convex. Its objective is concave
because the terms involving $t_{\rm sec}$, $\Delta$, and
$\operatorname{Tr}(\mathbf Z)$ are affine,
$-\lambda_b\mathcal B(\mathbf R_x)$ is concave, and the negative
proximal term is concave. Constraint
$\widetilde{\mathrm C}_3$ defines the hypograph of the concave rate
lower bound, while $\widetilde{\mathrm C}_4$ is an LMI. Equality in
$\widetilde{\mathrm C}_5$ is imposed coefficientwise in $\tau$ and
therefore represents finitely many affine equalities. Together with
$\widetilde{\mathrm C}_6$, it enforces the continuous
ghost-dominance requirement for every sampled Eve geometry. The
remaining constraints are affine or positive-semidefinite
constraints.

After convergence, if $\mathbf W_k$ is rank one, the corresponding
beamforming vector is recovered exactly as
\begin{equation}
    \mathbf w_k
    =
    \sqrt{\lambda_{\max}(\mathbf W_k)}
    \mathbf v_{\max}(\mathbf W_k),
    \label{eq:rank_one_recovery}
\end{equation}
where $\mathbf v_{\max}(\mathbf W_k)$ is a corresponding unit-norm principal eigenvector. 
For a higher-rank $\mathbf W_k$, we use the rank-one recovery
\begin{equation}
\mathbf w_k=
\frac{\mathbf W_k\mathbf h_{AB,k}}
{\sqrt{\mathbf h_{AB,k}^{H}\mathbf W_k\mathbf h_{AB,k}}},
\end{equation}
and transfer the residual covariance $\mathbf W_k-\mathbf w_k\mathbf w_k^H\succeq\mathbf 0$ to the deception covariance $\mathbf Z$. This preserves the total transmit covariance $\mathbf R_x$ and the desired received signal power at each Bob.

\subsection{SCA-based Solution Procedure}
\label{subsec:sca_solution}

Since the rate lower bounds in $\widetilde{\mathrm C}_3$ depend on the current covariance point,
problem~\eqref{eq:final_opt} is solved successively. The procedure starts from covariance matrices
$\{\mathbf W_k^{(0)}\}_{k=1}^{K}$ and $\mathbf Z^{(0)}$ for which the first
reformulated subproblem is feasible. At iteration $r$, the
interference powers $\{I_{B,k}^{(r)}\}_{k=1}^{K}$ are evaluated at
the current covariance point and substituted into
\eqref{eq:rate_lower_bound} to construct the concave lower bounds
$\{R_{B,k}^{\rm lb,(r)}\}_{k=1}^{K}$.

The complete variable set $\mathcal V$ is optimized when solving
problem~\eqref{eq:final_opt}. Let
$\mathcal V^{\star,(r)}$ denote its optimizer, and let
$\{\mathbf W_k^{\star,(r)}\}_{k=1}^{K}$ and
$\mathbf Z^{\star,(r)}$ denote the corresponding
covariance components. These components define the next SCA point as
\begin{equation}
    \mathbf W_k^{(r+1)}
    =
    \mathbf W_k^{\star,(r)},
    \quad \forall k,
    \qquad
    \mathbf Z^{(r+1)}
    =
    \mathbf Z^{\star,(r)}.
    \label{eq:sca_covariance_updates}
\end{equation}
The remaining variables in $\mathcal V$, including $t_{\rm sec}$,
$\Delta$, and the sum-of-squares (SOS) Gram matrices, are reoptimized at every
iteration and do not define the subsequent rate approximation or
proximal center. The procedure terminates when the relative change in
the objective of the relaxed reformulation, evaluated without the
proximal penalty, falls below a prescribed tolerance or when the
maximum number of iterations is reached. The complete procedure is summarized in Algorithm~\ref{alg:sca}.

\begin{algorithm}[t]
\caption{SCA-based Robust Covariance Design}
\label{alg:sca}
\begin{algorithmic}[1]

\Require Feasible initial covariances
$\{\mathbf W_k^{(0)}\}_{k=1}^{K}$ and $\mathbf Z^{(0)}$

\Require Reference values
$R_{\rm sec}^{\rm ref}$ and $\Delta^{\rm ref}$

\Require Weights
$\lambda_s$, $\lambda_g$, $\lambda_b$, and $\lambda_z$

\Require Proximal weight $\mu>0$

\Require Tolerance $\epsilon_{\rm SCA}>0$

\Require Maximum number of iterations $r_{\max}$

\State Set $r=0$

\Repeat

    \State Evaluate
    $\{I_{B,k}^{(r)}\}_{k=1}^{K}$
    at the current covariance point

    \State Construct
    $\{R_{B,k}^{\rm lb,(r)}\}_{k=1}^{K}$
    using \eqref{eq:rate_lower_bound}

    \State Solve problem~\eqref{eq:final_opt} to obtain
    $\mathcal V^{\star,(r)}$

    \State Extract
$\{\mathbf W_k^{\star,(r)}\}_{k=1}^{K}$ and
$\mathbf Z^{\star,(r)}$

\State Set
$\mathbf W_k^{(r+1)}
\leftarrow
\mathbf W_k^{\star,(r)}$,
$\forall k$

\State Set
$\mathbf Z^{(r+1)}
\leftarrow
\mathbf Z^{\star,(r)}$

    \State Evaluate the unregularized objective at
    $\mathcal V^{\star,(r)}$

    \State Set $r\leftarrow r+1$

\Until{the relative objective change is below
$\epsilon_{\rm SCA}$ or $r=r_{\max}$}

\Ensure Final covariance matrices
$\{\mathbf W_k^{(r)}\}_{k=1}^{K}$ and $\mathbf Z^{(r)}$

\end{algorithmic}
\end{algorithm}

The approximation $R_{B,k}^{\rm lb,(r)}$ is a global concave lower
bound on $R_{B,k}$ that is value-tight and gradient-consistent at the
current covariance point. Moreover, the proximal penalty and its
gradient vanish at that point. Therefore,
problem~\eqref{eq:final_opt} satisfies the local tightness and
first-order consistency properties required by the proximal SCA
framework. 

Assuming a feasible initialization, exact solution of each convex
subproblem, and the standard continuity and regularity conditions of
SCA, the objective values of the relaxed reformulation, evaluated
without the proximal term, form a nondecreasing sequence. The
transmit-power constraint bounds the covariance feasible set and,
consequently, the attainable objective values. Hence, the objective
sequence converges, and every accumulation point of the covariance
sequence is a stationary point of the semidefinite-relaxed
reformulation
\cite{razaviyayn2013,scutari2014}. After convergence, the information
beamformers are obtained through the rank-recovery procedure described
in Section~\ref{subsec:proposed_sca_subproblem}.

\subsection{Complexity Analysis}
\label{subsec:complexity}

The computational cost of Algorithm~\ref{alg:sca} is dominated by
the solution of the convex semidefinite-constrained
problem~\eqref{eq:final_opt} at each SCA iteration. Since
problem~\eqref{eq:final_opt} also contains logarithmic rate terms,
reciprocal-affine BCRB terms, and the quadratic proximal penalty, it
is not a pure SDP. Their precise treatment depends on the conic
reformulation and the adopted solver. We therefore characterize the
dominant complexity associated with the semidefinite cone structure,
which captures the main matrix operations of the proposed design.

The semidefinite part contains $K+1$ transmit-covariance PSD blocks
of size $N_t$, corresponding to
$\{\mathbf W_k\}_{k=1}^{K}$ and $\mathbf Z$, together with
$LKN_s$ robust Eve-decoding LMIs of size $N_e$. In addition, each
competing angular interval introduces two SOS Gram-matrix blocks of
sizes $N_e$ and $N_e-1$. Let
\begin{equation}
    N_I
    =
    \sum_{\ell=1}^{L}
    \sum_{i=1}^{N_s}
    |\mathcal S_{\ell i}|
    \leq
    2LN_s
\end{equation}
denote the total number of nonempty competing intervals.

Accordingly, the aggregate dimension of the semidefinite blocks is
\begin{equation}
\begin{aligned}
    \nu_{\rm PSD}
    &=
    (K+1)N_t
    +
    LKN_sN_e
    \\
    &+
    N_I(2N_e-1),
\end{aligned}
\label{eq:psd_barrier_dimension}
\end{equation}
while the dominant blockwise matrix-factorization cost of one
interior-point iteration scales as
\begin{equation}
\begin{aligned}
    \mathcal C_{\rm PSD}
    =
    \mathcal O\Big(
        &(K+1)N_t^3
        +
        LKN_sN_e^3
        \\
        &+
        N_I
        \big[
            N_e^3+(N_e-1)^3
        \big]
    \Big).
\end{aligned}
\label{eq:psd_iteration_complexity}
\end{equation}
The SOS coefficient-matching equalities are sparse, although the
exact Newton-system cost depends on how this sparsity and the
non-semidefinite convex terms are handled by the solver.

For a primal-dual interior-point implementation, the number of
Newton iterations associated with the semidefinite cones scales as
$\mathcal O(\sqrt{\nu_{\rm PSD}})$ up to logarithmic accuracy
factors. Hence, the dominant semidefinite contribution to one solution
of problem~\eqref{eq:final_opt} can be characterized as
\begin{equation}
    \mathcal O
    \left(
        \sqrt{\nu_{\rm PSD}}\,
        \mathcal C_{\rm PSD}
    \right).
\label{eq:sca_subproblem_complexity}
\end{equation}
For $I_{\rm SCA}$ outer iterations, the corresponding contribution
to the overall computational cost is
\begin{equation}
    \mathcal O
    \left(
        I_{\rm SCA}
        \sqrt{\nu_{\rm PSD}}\,
        \mathcal C_{\rm PSD}
    \right).
\label{eq:overall_complexity}
\end{equation}
In particular, for fixed $K$, $L$, $N_s$, and $N_e$, the dominant
semidefinite scaling with the number of transmit antennas reduces to
$\mathcal O(I_{\rm SCA}N_t^{3.5})$, up to solver-dependent and
logarithmic accuracy factors.

\section{Numerical Results}\label{sec:Num}

This section evaluates the proposed design in terms of secrecy, sensing deception, Alice-side Eve-angle estimation accuracy, and algorithm convergence. The results first verify controlled ghost formation for single- and multiple-Eve scenarios, and then compare the robust design with the considered benchmarks. The effects of sensing accuracy, angular uncertainty, and deception power are also examined. 

\subsection{Simulation Setup}\label{subsec:sim_setup}

Unless otherwise specified, Alice is placed at  $\mathbf p_A=(0,0)\,\mathrm{m}$, Bob at  $\mathbf p_{B,1} = (15,20)\,\mathrm{m}$,  and a single passive Eve at  $\mathbf p_{E,1} = (-15,15)\,\mathrm{m}$. Alice and the Eves employ half-wavelength-spaced ULAs, while the Bobs
are single-antenna users.

The Alice--Bob and Alice--Eve channels follow a Rician block-fading model:
\begin{align}
\mathbf h_{B,k}
&=
\sqrt{\beta_{AB,k}}
\Bigg(
\sqrt{\frac{K_{\rm ric}^{AB}}
{K_{\rm ric}^{AB}+1}}\,
\mathbf a_A(\theta_{AB,k})
\notag\\
&\qquad\qquad
+
\sqrt{\frac{1}
{K_{\rm ric}^{AB}+1}}\,
\mathbf g_{AB,k}
\Bigg),
\label{eq:sim_ab_channel}
\\
\mathbf H_{AE,\ell}
&=
\sqrt{\beta_{AE,\ell}}
\Bigg(
\sqrt{\frac{K_{\rm ric}^{AE}}
{K_{\rm ric}^{AE}+1}}\,
\mathbf a_E(\theta_{AE,\ell})
\mathbf a_A^H(\theta_{AE,\ell})
\notag\\
&\qquad\qquad
+
\sqrt{\frac{1}
{K_{\rm ric}^{AE}+1}}\,
\mathbf G_{AE,\ell}
\Bigg).
\label{eq:sim_ae_channel}
\end{align}
where $\beta_{AB,k}=\beta_0(d_0/d_{AB,k})^{\eta_{AB}}$ and $\beta_{AE,\ell}=\beta_0(d_0/d_{AE,\ell})^{\eta_{AE}}$. Moreover, $\mathbf g_{AB,k}\sim\mathcal{CN}(\mathbf0,\mathbf I_{N_t})$, while the entries of $\mathbf G_{AE,\ell}$ are i.i.d. $\mathcal{CN}(0,1)$. 
The NLoS realizations are drawn once per channel block and kept fixed across the candidate Eve angles within that block.
The passive-reflection coefficient from Bob $k$ to Eve $\ell$ is modeled as \cite{Musallam}
\begin{equation}
\rho_{\ell k}
=
\rho_0 \beta_{AB,k}
\left(
\frac{d_0}{d_{EB,\ell k}}
\right)^{\eta_{EB}}
\sigma_{\mathrm{RCS},k},
\end{equation}
where $\sigma_{\mathrm{RCS},k}$ denotes the effective RCS of the physical user associated with Bob $k$ and is modeled as log-normal with median $-4$~dBsm and standard deviation $4$~dB.

For Alice-side monostatic sensing, the average echo-power gain $\xi_{A,\ell}$ is modeled as
\begin{equation}
\xi_{A,\ell}
=
\xi_{A,0}
\left(
\frac{d_0}{d_{AE,\ell}}
\right)^{\eta_{A,\mathrm{tw}}}
\sigma_{\mathrm{RCS},E,\ell},
\end{equation}
where $\xi_{A,0}$ denotes the reference echo-power gain and $\eta_A^{\mathrm{tw}}$ is the two-way sensing path-loss exponent and $\sigma_{\mathrm{RCS},E,\ell}$ is the effective RCS of Eve $\ell$.

For the numerical evaluation, the angular prior of Eve $\ell$ is
modeled as a Gaussian-shaped density over the same uncertainty
sector $\Theta_{AE,\ell}$
adopted in Section~\ref{sec:BayesianUncertainty}, following the Bayesian angular-uncertainty
modeling in~\cite{Boz_iot}. The sector is centered at
the prior mean $\mu_{\ell}^{\rm prior}$ and defined as
\begin{equation}
    \theta_{\ell,1}
    =
    \mu_{\ell}^{\rm prior}
    -
    c_{\rm sup}\sigma_{\ell}^{\rm prior},
    \qquad
    \theta_{\ell,2}
    =
    \mu_{\ell}^{\rm prior}
    +
    c_{\rm sup}\sigma_{\ell}^{\rm prior},
\end{equation}
where $\sigma_{\ell}^{\rm prior}$ denotes the prior angular spread
and $c_{\rm sup}>0$ controls the sector width. The prior density is
normalized over this interval, and the corresponding prior Fisher
information $J_{P,\ell}$ is evaluated according to~\eqref{eq:prior_fisher_lemma}.

\begin{table}[t]
\caption{Simulation parameters.}
\label{tab:simulation_parameters}
\centering
\scriptsize
\renewcommand{\arraystretch}{1.05}
\setlength{\tabcolsep}{0.5pt}

\begin{tabularx}{\columnwidth}{
@{}
>{\raggedright\arraybackslash}p{0.235\columnwidth}
@{}
>{\centering\arraybackslash}p{0.215\columnwidth}
@{\hspace{7pt}}
>{\raggedright\arraybackslash}p{0.235\columnwidth}
@{}
>{\centering\arraybackslash}p{0.265\columnwidth}
@{}
}
\hline
\multicolumn{1}{l}{Parameter} &
\multicolumn{1}{l}{Value} &
\multicolumn{1}{l}{Parameter} &
\multicolumn{1}{l}{Value}
\\
\hline

$N_t,N_e$
& $8,4$
&
$K_{\rm ric}^{AB},K_{\rm ric}^{AE}$
& $5,5$
\\

$\eta_{AB},\eta_{AE},\eta_{EB}$
& $2.2,2.2,2.2$
&
$d_0,\beta_0,\rho_0$
& $1~{\rm m},1,1$
\\

$\sigma_b^2,\sigma_e^2,\sigma_{pE}^2$
& $10^{-6}$
&
$\sigma_{rA}^2$
& $10^{-4}$
\\

$\gamma_{B,\min},\Gamma_E$
& $1,0.63$
&
$L_A,N_s$
& $16,21$
\\

$\sigma_{\ell}^{\rm prior},c_{\rm sup}$
& $1.667^\circ,3$
&
$\Delta_g,\epsilon_g$
& $30^\circ,4^\circ$
\\

$\sigma_{\mathrm{RCS},E}$
& $1~\mathrm{m}^2$
&
MC realizations
& $200$
\\

$\mu$
& $5\times10^{-2}$
&
$\lambda_s,\lambda_g,\lambda_b,\lambda_z$
& $0.40,0.25,0.10,0.25$
\\

$\xi_{A,0}$                  & $10^{1.1}$ 
&
$\eta_A^{\mathrm{tw}}$       & $4$ \\

\hline
\end{tabularx}
\end{table}

The objective weights are selected to represent a secrecy-oriented operating point: $\lambda_s=0.40$ gives the largest emphasis to secrecy, while $\lambda_g=\lambda_z=0.25$ and $\lambda_b=0.10$ provide a balance between deception strength, deception power, and Alice-side sensing accuracy.

The transmit-power budget is varied as indicated in each figure and
is reported in dBm. All convex subproblems are implemented in MATLAB
and solved using CVX. Unless otherwise stated, the parameters in
Table~\ref{tab:simulation_parameters} are used. 

\subsection{Benchmark and Comparison Schemes}

To isolate the contributions of sensing privacy and robustness to eavesdropper uncertainty, three progressively enhanced secure ISAC schemes are considered.

For a fair comparison, all schemes use the same channel realizations, power budget, Bob QoS requirements, and Alice-side BCRB objective. The same weights and normalization factors are retained for all common objective terms. Secure ISAC (S-ISAC) omits only the ghost-dominance objective and constraint, while Secure ISAC with Sensing Privacy (S-ISAC-SP) and the proposed scheme use the same sensing-deception terms and differ in their treatment of Eve-angle
uncertainty.

\subsubsection{Secure ISAC (S-ISAC)}

The S-ISAC benchmark accounts for communication secrecy and Alice-side sensing performance, but does not impose any sensing-privacy requirement at Eve. The covariance design is performed using the nominal Eve geometry, and no ghost-dominance constraint or objective is included. This scheme therefore represents a conventional secure ISAC design in which security is characterized solely through Eve's data-decoding capability \cite{Su2024BCRB}.

\subsubsection{Secure ISAC with Sensing Privacy (S-ISAC-SP)}

The S-ISAC-SP scheme extends S-ISAC by additionally enforcing sensing privacy at Eve, following the general idea of sensing-security and controlled-deception designs   \cite{ Yang2026DualSecurity}. In particular, the ghost-dominance requirement is constructed using the nominal Eve geometry and its corresponding Alice--Eve channel, Bob--Eve bearing, reflected-power coefficient, and prescribed virtual direction. Hence, sensing privacy is promoted at the estimated Eve direction, without accounting for angular uncertainty.

\subsubsection{Secure ISAC with Sensing Privacy under Eavesdropper Uncertainty (Proposed)}

The proposed scheme further accounts for eavesdropper uncertainty by considering the complete set of geometry-consistent angular samples. For each candidate geometry, the Alice--Eve channel, Bob--Eve bearing, reflected-power coefficient, and prescribed virtual direction are jointly generated from the same Eve location. The information covariances and deception covariance are then jointly optimized under the robust Eve-decoding and ghost-dominance requirements over the sampled Eve geometries.

\subsection{Simulation Results}

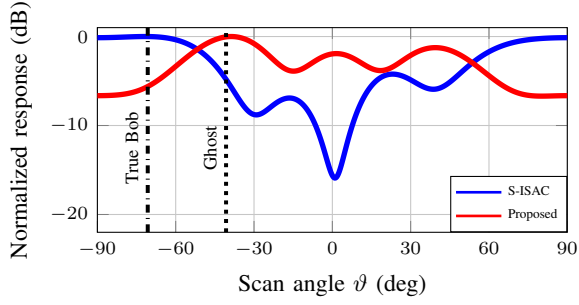
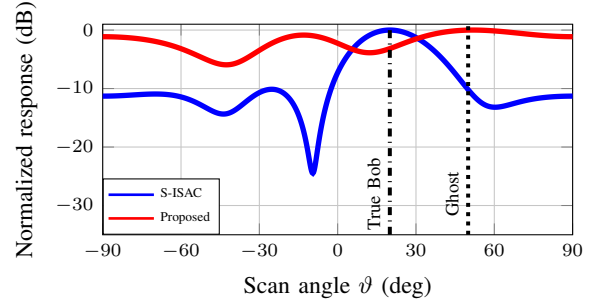
\begin{figure*}[t]
    \centering

    \subfloat[Passive response at Eve 1]{%
    \begin{minipage}{0.43\textwidth}
    \centering
    \begin{tikzpicture}
    \begin{axis}[
        width=\linewidth,
        height=0.55\linewidth,
        xlabel={Scan angle $\vartheta$ (deg)},
        ylabel={Normalized response (dB)},
        xmin=-90, xmax=90,
        ymin=-22, ymax=1,
        xtick={-90,-60,-30,0,30,60,90},
        grid=both,
        minor grid style={gray!20},
        major grid style={gray!45},
        tick label style={font=\scriptsize},
        label style={font=\small},
        title style={font=\small},
        legend columns=1,
        legend cell align=left,
        legend style={
            font=\tiny,
            at={(1,0.00)},
            anchor=south east,
            draw=black,
            fill=white,
            inner sep=1.5pt
        },
        legend image post style={line width=1pt},
    ]

    \addplot[blue, line width=2pt]
    table[x=theta, y=BaselineA, col sep=space]
    {Figures3/passive_beampattern_2eve_eve1.dat};
    \addlegendentry{S-ISAC}


    \addplot[red, line width=2pt, solid]
    table[x=theta, y=Proposed, col sep=space]
    {Figures3/passive_beampattern_2eve_eve1.dat};
    \addlegendentry{Proposed}

    \addplot[black, dash dot, line width=1.5pt]
    coordinates {(-70.71,-22) (-70.71,1)};
    
    \addplot[black, dotted, line width=1.5pt]
    coordinates {(-40.71,-22) (-40.71,1)};
    
    \node[font=\scriptsize, rotate=90, anchor=south]
    at (axis cs:-70.71,-12.5) {True Bob};
    
    \node[font=\scriptsize, rotate=90, anchor=south]
    at (axis cs:-40.71,-12.5) {Ghost};

    \end{axis}
    \end{tikzpicture}
    \end{minipage}
    }
    \hspace{0.025\textwidth}
    \subfloat[Passive response at Eve 2]{%
    \begin{minipage}{0.43\textwidth}
    \centering
    \begin{tikzpicture}
    \begin{axis}[
        width=\linewidth,
        height=0.56\linewidth,
        xlabel={Scan angle $\vartheta$ (deg)},
        ylabel={Normalized response (dB)},
        xmin=-90, xmax=90,
        ymin=-35, ymax=1,
        xtick={-90,-60,-30,0,30,60,90},
        grid=both,
        minor grid style={gray!20},
        major grid style={gray!45},
        tick label style={font=\scriptsize},
        label style={font=\small},
        title style={font=\small},
        legend columns=1,
        legend cell align=left,
        legend style={
            font=\tiny,
            at={(0.00,0.00)},
            anchor=south west,
            draw=black,
            fill=white,
            inner sep=1.5pt
        },
        legend image post style={line width=1pt},
    ]

    \addplot[blue, line width=2pt]
    table[x=theta, y=BaselineA, col sep=space]
    {Figures3/passive_beampattern_2eve_eve2.dat};
    \addlegendentry{S-ISAC}


    \addplot[red, line width=2pt, solid]
    table[x=theta, y=Proposed, col sep=space]
    {Figures3/passive_beampattern_2eve_eve2.dat};
    \addlegendentry{Proposed}

    \addplot[black, dash dot, line width=1.5pt]
    coordinates {(19.98,-35) (19.98,1)};
    
    \addplot[black, dotted, line width=1.5pt]
    coordinates {(49.98,-35) (49.98,1)};
    
    \node[font=\scriptsize, rotate=90, anchor=south]
    at (axis cs:19.98,-28) {True Bob};
    
    \node[font=\scriptsize, rotate=90, anchor=south]
    at (axis cs:49.98,-28) {Ghost};

    \end{axis}
    \end{tikzpicture}
    \end{minipage}
    }

    \caption{Passive angular responses for the two-Eve case with
    $P_{\max}=30$ dBm.}
    \label{fig:beam_two_eves}
\end{figure*}

Fig.~\ref{fig:beam_two_eves} evaluates the ability of the transmit covariance to control the dominant direction observed by the Eves' passive angular sensing in an asymmetric two-Eve configuration, where $\mathbf p_{B,1}=(25,-12)\,\mathrm{m}$, $\mathbf p_{E,1}=(18,8)\,\mathrm{m}$, and $\mathbf p_{E,2}=(14,-16)\,\mathrm{m}$, with a prescribed ghost angular offset of $\Delta_g=30^\circ$. The normalized passive angular response of each Eve is shown over the scan-angle range, highlighting the dominant direction inferred through passive angular sensing. Under S-ISAC, the dominant passive response remains aligned with the true Bob direction, allowing each Eve to infer the corresponding angular information. In contrast, the proposed scheme shifts the dominant response toward the prescribed ghost direction, thereby deceiving both Eves with the intended angular offset. This demonstrates the effectiveness of the proposed sensing-privacy mechanism and its ability to simultaneously protect against multiple Eves.

\begin{figure}[t]
\centering
\begin{tikzpicture}
\begin{axis}[
    width=0.75\linewidth,
    height=0.65\linewidth,
    xlabel={Transmit power $P_{\max}$ (dBm)},
    ylabel={Average secrecy rate (bps/Hz)},
    xmin=25, xmax=50,
    ymin=5, ymax=15,
    xtick={25,30,35,40,45,50},
    ytick={5,6,7,8,9,10,11,12,13,14,15},
    grid=both,
    minor grid style={gray!20},
    major grid style={gray!45},
    tick label style={font=\scriptsize},
    label style={font=\small},
    legend columns=1,
    legend cell align=left,
    legend style={
        font=\scriptsize,
        at={(0,1)},
        anchor=north west,
        draw=black,
        fill=white,
        inner sep=1.5pt,
        row sep=1pt,
        column sep=4pt
    },
    legend image post style={line width=1pt},
]

\addplot[
    blue,
    mark=o,
    mark options={solid},
    mark size=2pt,
    line width=2pt
]
table[
    x=Pmax_dBm,
    y=S_ISAC_SRavg,
    col sep=space
]
{Figures3/average_SR_under_angular_uncertainty_vs_power.dat};
\addlegendentry{S-ISAC}

\addplot[
    black,
    mark=diamond,
    mark options={solid},
    mark size=2pt,
    line width=2pt
]
table[
    x=Pmax_dBm,
    y=S_ISAC_SP_SRavg,
    col sep=space
]
{Figures3/average_SR_under_angular_uncertainty_vs_power.dat};
\addlegendentry{S-ISAC-SP}

\addplot[
    red,
    mark=triangle,
    mark options={solid},
    mark size=2pt,
    line width=2pt
]
table[
    x=Pmax_dBm,
    y=Proposed_SRavg,
    col sep=space
]
{Figures3/average_SR_under_angular_uncertainty_vs_power.dat};
\addlegendentry{Proposed}

\end{axis}
\end{tikzpicture}
\caption{Average SR under angular uncertainty versus the transmit-power budget for the single-Bob single-Eve setup.}
\label{fig:sr_avg_power}
\end{figure}
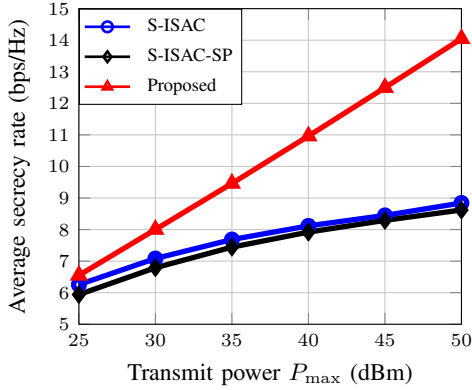

Fig.~\ref{fig:sr_avg_power} shows the average SR versus Alice's transmit-power budget for the single-Bob single-Eve configuration under Eve-angle uncertainty. The SR increases with $P_{\max}$ for all schemes, as the larger power budget provides greater flexibility to strengthen Bob's useful signal while limiting Eve's decoding capability. S-ISAC and S-ISAC-SP exhibit similar performance, with S-ISAC-SP incurring a small SR penalty due to the additional resources devoted to sensing privacy. However, both schemes are designed for the nominal Eve geometry and their security performance deteriorates as the actual Eve direction deviates from this nominal point. In contrast, the proposed scheme explicitly accounts for the complete angular uncertainty sector and therefore maintains effective Eve-decoding protection across the possible Eve directions. Consequently, it achieves a substantially higher average SR, with the performance gap becoming more pronounced as $P_{\max}$ increases. This demonstrates the benefit of incorporating Eve-angle uncertainty directly into the joint communication-security and sensing-privacy design.

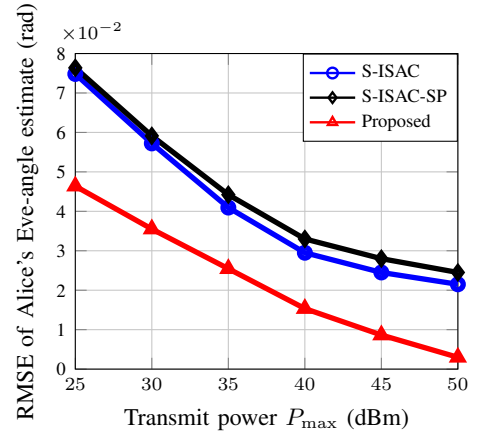
\begin{figure}[t]
\centering
\begin{tikzpicture}
\begin{axis}[
    width=0.75\linewidth,
    height=0.65\linewidth,
    xlabel={Transmit power $P_{\max}$ (dBm)},
    ylabel={RMSE of Alice's Eve-angle estimate (rad)},
    xmin=25, xmax=50,
    ymin=0, ymax=0.08,
    xtick={25,30,35,40,45,50},
    ytick={0,0.01,0.02,0.03,0.04,0.05,0.06,0.07,0.08},
    yticklabels={0,1,2,3,4,5,6,7,8},
    scaled y ticks=false,
    extra description/.code={
        \node[anchor=south east,font=\scriptsize]
        at (rel axis cs:0.15,1.01) {$\times 10^{-2}$};
    }, 
    grid=both,
    minor grid style={gray!20},
    major grid style={gray!45},
    tick label style={font=\scriptsize},
    label style={font=\small},
    legend columns=1,
    legend cell align=left,
    legend style={
        font=\scriptsize,
        at={(1,1)},
        anchor=north east,
        draw=black,
        fill=white,
        inner sep=2pt,
        row sep=-1pt
    },
    legend image post style={line width=0.9pt},
]

\addplot[
    blue,
    mark=o,
    mark options={solid},
    mark size=2pt,
    line width=2pt
]
table[
    x=Pmax_dBm,
    y=S_ISAC_RMSE_rad,
    col sep=space
]
{Figures3/alice_angle_RMSE_vs_power_rad.dat};
\addlegendentry{S-ISAC}

\addplot[
    black,
    mark=diamond,
    mark options={solid},
    mark size=2.0pt,
    line width=2pt
]
table[
    x=Pmax_dBm,
    y=S_ISAC_SP_RMSE_rad,
    col sep=space
]
{Figures3/alice_angle_RMSE_vs_power_rad.dat};
\addlegendentry{S-ISAC-SP}

\addplot[
    red,
    mark=triangle,
    mark options={solid},
    mark size=2.0pt,
    line width=2pt
]
table[
    x=Pmax_dBm,
    y=Proposed_RMSE_rad,
    col sep=space
]
{Figures3/alice_angle_RMSE_vs_power_rad.dat};
\addlegendentry{Proposed}

\end{axis}
\end{tikzpicture}
\caption{Empirical RMSE of Alice's Eve-angle estimate versus the transmit-power budget.}
\label{fig:rmse_power}
\end{figure}

Fig.~\ref{fig:rmse_power} reports the empirical RMSE of Alice's Eve-angle estimate versus the transmit-power budget. 
The estimation error decreases monotonically with $P_{\max}$ for all schemes, since the additional transmit power improves the sensing information available at Alice. S-ISAC and S-ISAC-SP achieve comparable performance, although the latter exhibits a moderately larger RMSE because part of the available spatial and power resources must also support the prescribed sensing-deception response. The proposed scheme consistently achieves the lowest RMSE over the entire power range. By explicitly accounting for the possible Eve geometries in the covariance design, it produces a transmit covariance that remains informative for Eve-angle estimation throughout the uncertainty sector rather than being tailored primarily to the nominal geometry. The gain becomes increasingly pronounced as $P_{\max}$ grows, indicating that the proposed design can exploit the additional transmit resources to improve Alice's localization capability while simultaneously maintaining robust communication security and sensing privacy.

\begin{figure}[t]
\centering

\begin{minipage}[t]{0.485\linewidth}
    \centering
    \begin{tikzpicture}
    \begin{axis}[
        width=1.00\linewidth,
        height=1.3\linewidth,
        xlabel={SCA iteration},
        ylabel={Secrecy rate lower bound (bps/Hz)},
        xmin=0, xmax=5,
        ymin=0, ymax=12,
        xtick={0,1,2,3,4,5},
        ytick={0,2,4,6,8,10,12},
        scaled y ticks=false,
        grid=both,
        minor grid style={gray!20},
        major grid style={gray!45},
        tick label style={font=\tiny},
        label style={font=\scriptsize},
        legend columns=1,
        legend cell align=left,
        legend style={
            font=\tiny,
            at={(1,0)},
            anchor=south east,
            draw=black,
            fill=white,
            inner sep=1pt,
            row sep=-1pt
        },
        legend image post style={line width=0.8pt},
    ]

    \addplot[
        blue,
        mark=o,
        mark options={solid},
        mark size=1.5pt,
        line width=1.5pt
    ]
    table[
        x=iter,
        y=SR_25dBm,
        col sep=space,
        comment chars={\#}
    ]
    {Figures3/sca_convergence_3curves.dat};
    \addlegendentry{$P_{\max}=25$ dBm}

    \addplot[
        red,
        mark=square,
        mark options={solid},
        mark size=1.5pt,
        line width=1.5pt
    ]
    table[
        x=iter,
        y=SR_30dBm,
        col sep=space,
        comment chars={\#}
    ]
    {Figures3/sca_convergence_3curves.dat};
    \addlegendentry{$P_{\max}=30$ dBm}

    \addplot[
        black,
        mark=triangle,
        mark options={solid},
        mark size=1.5pt,
        line width=1.5pt
    ]
    table[
        x=iter,
        y=SR_40dBm,
        col sep=space,
        comment chars={\#}
    ]
    {Figures3/sca_convergence_3curves.dat};
    \addlegendentry{$P_{\max}=40$ dBm}

    \end{axis}
    \end{tikzpicture}

    \vspace{-1.5mm}
    {\scriptsize (a)}
\end{minipage}
\hfill
\begin{minipage}[t]{0.485\linewidth}
    \centering
    \begin{tikzpicture}
    \begin{axis}[
        width=1.00\linewidth,
        height=1.30\linewidth,
        xlabel={SCA iteration},
        ylabel={Average Alice root-BCRB (rad)},
        xmin=0, xmax=5,
        ymin=0.01, ymax=0.032,
        xtick={0,1,2,3,4,5},
        ytick={0.01,0.015,0.02,0.025,0.03},
        yticklabels={1,1.5,2,2.5,3},
        scaled y ticks=false,
        extra description/.code={
            \node[anchor=south west,font=\scriptsize]
            at (rel axis cs:-0.1,1.01) {$\times 10^{-2}$};
        },
        grid=both,
        minor grid style={gray!20},
        major grid style={gray!45},
        tick label style={font=\tiny},
        label style={font=\scriptsize},
        legend columns=1,
        legend cell align=left,
        legend style={
            font=\tiny,
            at={(1,0.7)},
            anchor=north east,
            draw=black,
            fill=white,
            inner sep=1pt,
            row sep=-1pt
        },
        legend image post style={line width=0.8pt},
    ]

    \addplot[
        blue,
        mark=o,
        mark options={solid},
        mark size=1.5pt,
        line width=1.5pt
    ]
    table[
        x=iter,
        y expr={\thisrow{BCRB_25dBm}*0.0174532925199433},
        col sep=space,
        comment chars={\#}
    ]
    {Figures3/sca_convergence_3curves.dat};
    \addlegendentry{$P_{\max}=25$ dBm}

    \addplot[
        red,
        mark=square,
        mark options={solid},
        mark size=1.5pt,
        line width=1.5pt
    ]
    table[
        x=iter,
        y expr={\thisrow{BCRB_30dBm}*0.0174532925199433},
        col sep=space,
        comment chars={\#}
    ]
    {Figures3/sca_convergence_3curves.dat};
    \addlegendentry{$P_{\max}=30$ dBm}

    \addplot[
        black,
        mark=triangle,
        mark options={solid},
        mark size=1.5pt,
        line width=1.5pt
    ]
    table[
        x=iter,
        y expr={\thisrow{BCRB_40dBm}*0.0174532925199433},
        col sep=space,
        comment chars={\#}
    ]
    {Figures3/sca_convergence_3curves.dat};
    \addlegendentry{$P_{\max}=40$ dBm}

    \end{axis}
    \end{tikzpicture}

    \vspace{-1.5mm}
    {\scriptsize (b)}
\end{minipage}

\caption{Convergence of the proposed SCA algorithm in terms of secrecy rate lower bound and Alice root-BCRB.}
\label{fig:sca_convergence}
\end{figure}
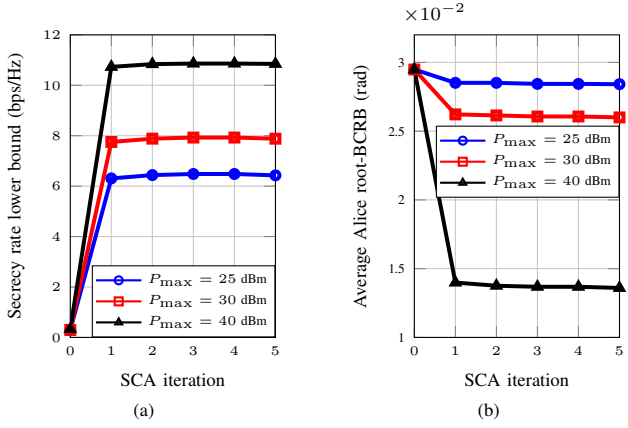

Fig.~\ref{fig:sca_convergence} illustrates the convergence
behavior of the proposed SCA algorithm for various transmit-power
budgets. Starting from the feasible initialization, the SR
lower bound increases sharply after the first SCA update, while the
root-BCRB of Alice's Eve-angle estimate decreases accordingly. The
changes are more pronounced at higher $P_{\max}$, where the larger
power budget provides greater flexibility for the covariance update
to improve the communication-security and sensing objectives relative
to the initial feasible point. After this initial adjustment, both
metrics rapidly stabilize, with only marginal changes in the subsequent
iterations. As expected, increasing $P_{\max}$ leads to a higher
converged SR lower bound and a lower root-BCRB. Overall, the
results show that the proposed SCA procedure reaches a stable operating
point within only a few iterations across the considered power levels.

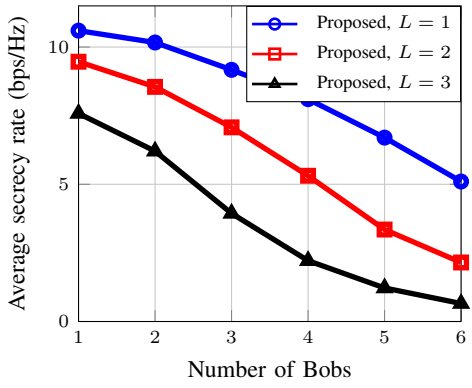
\begin{figure}[t]
    \centering
    \begin{tikzpicture}
    \begin{axis}[
        width=0.75\linewidth,
        height=0.65\linewidth,
        xlabel={Number of Bobs},
        ylabel={Average secrecy rate (bps/Hz)},
        xmin=1, xmax=6,
        ymin=0, ymax=11.5,
        xtick={1,2,3,4,5,6},
        scaled y ticks=false,
        grid=both,
        minor grid style={gray!20},
        major grid style={gray!45},
        tick label style={font=\scriptsize},
        label style={font=\small},
        legend columns=1,
        legend cell align=left,
        legend style={
            font=\scriptsize,
            at={(1,1)},
            anchor=north east,
            draw=black,
            fill=white,
            inner sep=1.5pt,
            row sep=1pt,
            column sep=5pt
        },
        legend image post style={line width=1pt},
    ]

    \addplot[
        blue,
        mark=o,
        mark options={solid},
        mark size=2.0,
        line width=2pt,
        solid
    ]
    table[x=K, y=Proposed_1E, col sep=space]
    {Figures3/sr-Bobs_new.dat};
    \addlegendentry{Proposed, $L=1$}

    \addplot[
        red,
        mark=square,
        mark options={solid},
        mark size=2.0,
        line width=2pt,
        solid
    ]
    table[x=K, y=Proposed_2E, col sep=space]
    {Figures3/sr-Bobs_new.dat};
    \addlegendentry{Proposed, $L=2$}

    \addplot[
        black,
        mark=triangle,
        mark options={solid},
        mark size=2,
        line width=2pt,
        solid
    ]
    table[x=K, y=Proposed_3E, col sep=space]
    {Figures3/sr-Bobs_new.dat};
    \addlegendentry{Proposed, $L=3$}



    \end{axis}
    \end{tikzpicture}
   \caption{Average secrecy rate versus the number of Bobs for different numbers of Eves with $P_{\max}=40$~dBm.}
    \label{fig:sr_vs_bobs}
\end{figure}

In Fig.~\ref{fig:sr_vs_bobs}, the average worst-user SR is evaluated versus the number of served Bobs for different numbers of Eves, namely $L=1$, $2$, and $3$, with $P_{\max}=40$~dBm. As expected, the SR decreases as the number of Bobs increases for all considered values of $L$. Serving additional Bobs requires the fixed transmit-power budget and spatial degrees of freedom to be shared among more information streams, while simultaneously increasing inter-user interference and the number of secrecy constraints. The performance degradation becomes more pronounced as the number of Eves increases, since each additional Eve introduces further robust Eve-decoding constraints across all Bobs and an additional ghost-dominance requirement for its selected Bob.
Hence, the design becomes progressively more constrained as both $K$ and $L$ increase, explaining the steeper degradation observed for $L=2$ and particularly for $L=3$. Nevertheless, for a single Eve, the proposed scheme retains a substantial SR even under the most heavily loaded configuration considered, remaining above $5$~bps/Hz for $K=6$. This indicates that the proposed robust design can support a relatively large number of legitimate users with satisfactory secrecy performance when the number of Eves is moderate.

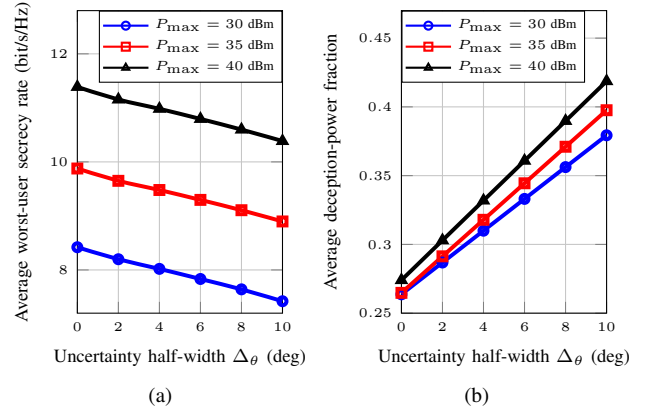
\begin{figure}[t]
    \centering

    \subfloat[\label{fig:uncertainty_sr}]{
    \begin{minipage}[t]{0.485\linewidth}
        \centering
        \begin{tikzpicture}
        \begin{axis}[
            width=1.0\linewidth,
            height=1.30\linewidth,
            xlabel={Uncertainty half-width $\Delta_{\theta}$ (deg)},
            ylabel={Average worst-user secrecy rate (bit/s/Hz)},
            xmin=0, xmax=10,
            ymin=7.2, ymax=12.8,
            xtick={0,2,4,6,8,10},
            scaled y ticks=false,
            grid=both,
            minor grid style={gray!20},
            major grid style={gray!45},
            tick label style={font=\tiny},
            label style={font=\scriptsize},
            legend columns=1,
            legend cell align=left,
            legend style={
                font=\tiny,
                at={(1,1)},
                anchor=north east,
                draw=black,
                fill=white,
                inner sep=1pt,
                row sep=-1pt
            },
            legend image post style={line width=0.8pt},
        ]

        \addplot[
            blue,
            mark=o,
            mark options={solid},
            mark size=1.5pt,
            line width=1.5pt,
            smooth,
            tension=0.55
        ]
        table[
            x=Udeg,
            y=SRwc_30dBm,
            col sep=space
        ]
        {Figures3/sr_uncertainty.dat};
        \addlegendentry{$P_{\max}=30$ dBm}

        \addplot[
            red,
            mark=square,
            mark options={solid},
            mark size=1.5pt,
            line width=1.5pt,
            smooth,
            tension=0.55
        ]
        table[
            x=Udeg,
            y=SRwc_35dBm,
            col sep=space
        ]
        {Figures3/sr_uncertainty.dat};
        \addlegendentry{$P_{\max}=35$ dBm}

        \addplot[
            black,
            mark=triangle,
            mark options={solid},
            mark size=1.5pt,
            line width=1.5pt,
            smooth,
            tension=0.55
        ]
        table[
            x=Udeg,
            y=SRwc_40dBm,
            col sep=space
        ]
        {Figures3/sr_uncertainty.dat};
        \addlegendentry{$P_{\max}=40$ dBm}

        \end{axis}
        \end{tikzpicture}
    \end{minipage}
    }
    \hspace{-0.04\linewidth}%
    \subfloat[\label{fig:uncertainty_power}]{
    \begin{minipage}[t]{0.485\linewidth}
        \centering
        \begin{tikzpicture}
        \begin{axis}[
            width=1.0\linewidth,
            height=1.30\linewidth,
            xlabel={Uncertainty half-width $\Delta_{\theta}$ (deg)},
            ylabel={Average deception-power fraction},
            xmin=0, xmax=10,
            ymin=0.25, ymax=0.47,
            xtick={0,2,4,6,8,10},
            scaled y ticks=false,
            grid=both,
            minor grid style={gray!20},
            major grid style={gray!45},
            tick label style={font=\tiny},
            label style={font=\scriptsize},
            legend columns=1,
            legend cell align=left,
            legend style={
                font=\tiny,
                at={(0,1)},
                anchor=north west,
                draw=black,
                fill=white,
                inner sep=1pt,
                row sep=-1pt
            },
            legend image post style={line width=0.8pt},
        ]

        \addplot[
            blue,
            mark=o,
            mark options={solid},
            mark size=1.5pt,
            line width=1.5pt,
            smooth,
            tension=0.55
        ]
        table[
            x=Udeg,
            y=Zfrac_30dBm,
            col sep=space
        ]
        {Figures3/sr_uncertainty.dat};
        \addlegendentry{$P_{\max}=30$ dBm}

        \addplot[
            red,
            mark=square,
            mark options={solid},
            mark size=1.5pt,
            line width=1.5pt,
            smooth,
            tension=0.55
        ]
        table[
            x=Udeg,
            y=Zfrac_35dBm,
            col sep=space
        ]
        {Figures3/sr_uncertainty.dat};
        \addlegendentry{$P_{\max}=35$ dBm}

        \addplot[
            black,
            mark=triangle,
            mark options={solid},
            mark size=1.5pt,
            line width=1.5pt,
            smooth,
            tension=0.55
        ]
        table[
            x=Udeg,
            y=Zfrac_40dBm,
            col sep=space
        ]
        {Figures3/sr_uncertainty.dat};
        \addlegendentry{$P_{\max}=40$ dBm}

        \end{axis}
        \end{tikzpicture}
    \end{minipage}
    }
    
    \caption{Average worst-user secrecy rate and deception-power fraction versus the Eve-angle uncertainty half-width.}
    \label{fig:uncertainty_sr_deception_power}
\end{figure}

In Fig.~\ref{fig:uncertainty_sr}, the average worst-user SR is evaluated versus the Eve-angle uncertainty half-width $\Delta_{\theta}$ for different transmit-power budgets. As $\Delta_{\theta}$ increases, the SR gradually decreases, since the proposed design must satisfy the communication-security and sensing-privacy requirements over a broader range of possible Eve geometries. This expected degradation remains relatively moderate over the considered uncertainty range, demonstrating the robustness of the proposed design to imperfect Eve-angle knowledge. From $\Delta_{\theta}=0^\circ$ to $10^\circ$, the average SR decreases by approximately $11.87\%$, $9.92\%$, and $8.80\%$ for $P_{\max}=30$, $35$, and $40$~dBm, respectively. The smaller relative loss at higher transmit powers indicates that additional transmit resources provide greater flexibility to accommodate angular uncertainty while preserving secrecy performance. In Fig.~\ref{fig:uncertainty_power}, the average deception-power fraction, defined as $\zeta_Z=\operatorname{Tr}(\mathbf Z)/P_{\max}$, is evaluated versus the Eve-angle uncertainty width $\Delta_{\theta}$. The deception-power fraction increases consistently with $\Delta_{\theta}$ for all considered transmit-power budgets, since a broader uncertainty sector requires the prescribed deceptive response to be maintained over a wider range of possible Eve geometries. Nevertheless, $\zeta_Z$ remains below approximately $0.42$ throughout the considered range, indicating that less than half of the available power budget is allocated to deception even under the largest uncertainty. 

\begin{figure}[t]
\centering
\begin{tikzpicture}
\begin{axis}[
    width=0.75\linewidth,
    height=0.65\linewidth,
    xlabel={Deception weight $\lambda_g$},
    ylabel={Average secrecy rate (bit/s/Hz)},
    xmin=0.10, xmax=0.50,
    ymin=5, ymax=12,
    xtick={0.10,0.15,0.20,0.25,0.30,0.35,0.40,0.45,0.50},
    ytick={5,6,7,8,9,10,11,12},
    grid=both,
    minor grid style={gray!20},
    major grid style={gray!45},
    tick label style={font=\scriptsize},
    label style={font=\small},
    legend columns=1,
    legend cell align=left,
    legend style={
        font=\scriptsize,
        at={(1,1)},
        anchor=north east,
        draw=black,
        fill=white,
        inner sep=2pt,
        row sep=-1pt
    },
    legend image post style={line width=0.9pt},
]

\addplot[
    blue,
    mark=o,
    mark options={solid},
    mark size=2pt,
    line width=2pt
]
table[
    x=lambda_g,
    y=SR_P30,
    col sep=space
]
{Figures3/lambda_weight_sr.dat};
\addlegendentry{$P_{\max}=30$ dBm}

\addplot[
    black,
    mark=diamond,
    mark options={solid},
    mark size=2pt,
    line width=2pt
]
table[
    x=lambda_g,
    y=SR_P35,
    col sep=space
]
{Figures3/lambda_weight_sr.dat};
\addlegendentry{$P_{\max}=35$ dBm}

\addplot[
    red,
    mark=triangle,
    mark options={solid},
    mark size=2pt,
    line width=2pt
]
table[
    x=lambda_g,
    y=SR_P40,
    col sep=space
]
{Figures3/lambda_weight_sr.dat};
\addlegendentry{$P_{\max}=40$ dBm}

\end{axis}
\end{tikzpicture}
\caption{Average SR versus the deception weight $\lambda_g$, with $\lambda_s=0.65-\lambda_g$.}
\label{fig:weight_sweep}
\end{figure}
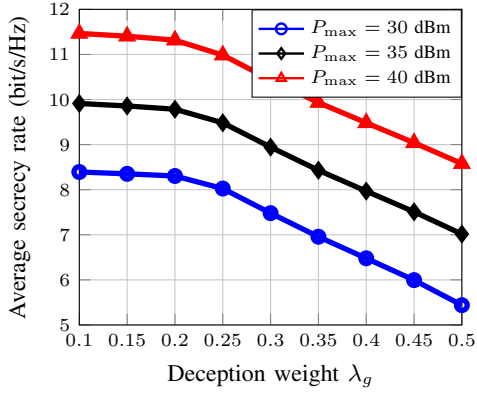

Fig.~\ref{fig:weight_sweep} investigates the sensitivity of the proposed scheme to the objective weights by varying $\lambda_g$ while fixing $\lambda_b=0.10$, $\lambda_z=0.25$, and setting $\lambda_s=0.65-\lambda_g$. The average SR decreases monotonically with $\lambda_g$ for all transmit-power budgets, since greater emphasis is placed on deception while the secrecy weight is reduced. Increasing $P_{\max}$ consistently improves the secrecy rate, reflecting the greater flexibility provided by the additional transmit resources. The adopted setting $\lambda_g=0.25$ and $\lambda_s=0.40$ lies near the transition region and provides a balanced operating point.

\section{Conclusion}\label{sec:Conc}

This paper investigated sensing privacy in secure ISAC under uncertainty in the Eves' angular directions. A robust multi-Eve covariance-design framework was developed to jointly protect confidential communication and manipulate the spatial information inferred through passive sensing toward prescribed deceptive directions. The proposed design jointly accounts for communication security,
Alice's sensing accuracy, sensing privacy at the Eves, and Eve-angle
uncertainty within a unified optimization framework. Numerical results showed that conventional secure-ISAC designs do not necessarily ensure sensing privacy and that ignoring Eve-angle uncertainty can significantly degrade both secrecy and deception performance. In contrast, the proposed scheme consistently outperformed the benchmark designs, maintained controlled sensing deception for single- and multiple-Eve scenarios, and remained effective even for higher angular uncertainties.

\appendices 

\appendices
\section{Proof of Lemma~\ref{lem:eve_bcrb}}
\label{app:bcrb_proof}

From the monostatic sensing model in
\eqref{eq:alice_sensing_model}, the conditional mean of Alice's
sensing observation for Eve $\ell$ at snapshot $q$ is
\begin{equation}
    \boldsymbol{\mu}_{A,\ell}
    [q;\theta,\alpha_{A,\ell}]
    =
    \alpha_{A,\ell}
    \mathbf A(\theta)
    \mathbf x[q].
\end{equation}
Its derivative with respect to the angular parameter is
\begin{equation}
    \frac{
        \partial
        \boldsymbol{\mu}_{A,\ell}
        [q;\theta,\alpha_{A,\ell}]
    }{
        \partial\theta
    }
    =
    \alpha_{A,\ell}
    \dot{\mathbf A}(\theta)
    \mathbf x[q].
\end{equation}

Given the transmitted sensing snapshots, the classical Fisher
information matrix associated with $\boldsymbol{\eta}_{\ell}$ has the block structure
\begin{equation}
    \mathbf J_{\ell}
    \left(
        \boldsymbol{\eta}_{\ell}
    \right)
    =
    \begin{bmatrix}
        J_{\theta\theta,\ell}
        &
        \mathbf J_{\theta\alpha,\ell}
        \\
        \mathbf J_{\theta\alpha,\ell}^{T}
        &
        \mathbf J_{\alpha\alpha,\ell}
    \end{bmatrix},
\end{equation}
where $\mathbf J_{\theta\alpha,\ell}$ contains the
cross-information terms between the angular parameter and the real
and imaginary parts of $\alpha_{A,\ell}$.

For independent circularly symmetric complex Gaussian sensing noise
with covariance $\sigma_{rA}^2\mathbf I_{N_t}$, the angle--angle
entry is
\begin{align}
    J_{\theta\theta,\ell}
    \left(
        \theta,\alpha_{A,\ell}
    \right)
    &=
    \frac{2}{\sigma_{rA}^2}
    \sum_{q=1}^{L_A}
    \Re
    \left\{
        \left(
            \frac{
                \partial
                \boldsymbol{\mu}_{A,\ell}[q]
            }{
                \partial\theta
            }
        \right)^H
        \left(
            \frac{
                \partial
                \boldsymbol{\mu}_{A,\ell}[q]
            }{
                \partial\theta
            }
        \right)
    \right\}
    \nonumber\\
    &=
    \frac{
        2|\alpha_{A,\ell}|^2
    }{
        \sigma_{rA}^2
    }
    \sum_{q=1}^{L_A}
    \mathbf x^H[q]
    \mathbf Q_A(\theta)
    \mathbf x[q].
    \label{eq:app_Jtt}
\end{align}
Since $\mathbf Q_A(\theta)$ is Hermitian positive semidefinite, the
quadratic form in \eqref{eq:app_Jtt} is real and nonnegative. Using the cyclic property of the trace operator, the quadratic form
in \eqref{eq:app_Jtt} can be written as
\begin{equation}
    \mathbf x^H[q]
    \mathbf Q_A(\theta)
    \mathbf x[q]
    =
    \operatorname{Tr}
    \left(
        \mathbf Q_A(\theta)
        \mathbf x[q]\mathbf x^H[q]
    \right).
\end{equation}

 Therefore, averaging
\eqref{eq:app_Jtt} over the transmitted snapshots, the reflection
coefficient, and the angular prior yields
\begin{align}
    \overline J_{D,\ell}(\mathbf R_x)
    &=
    \mathbb E_{\theta,\alpha_{A,\ell},\mathbf x}
    \left\{
        J_{\theta\theta,\ell}
        \left(
            \theta,\alpha_{A,\ell}
        \right)
    \right\}
    \nonumber\\
    &=
    \int_{\theta_{\ell,1}}^{\theta_{\ell,2}}
    \frac{
        2L_A\xi_{A,\ell}
    }{
        \sigma_{rA}^2
    }
    \operatorname{Tr}
    \left(
        \mathbf Q_A(\theta)
        \mathbf R_x
    \right)
    p_\ell(\theta)
    \,d\theta
    \nonumber\\
    &=
    \operatorname{Tr}
    \left(
        \mathbf Q_{B,\ell}
        \mathbf R_x
    \right),
    \label{eq:app_average_data_information}
\end{align}
where $\mathbf Q_{B,\ell}$ is defined in
\eqref{eq:QB_lemma}.

The data cross-information terms between
$\theta_{AE,\ell}$ and the real and imaginary parts of
$\alpha_{A,\ell}$ contain factors proportional to
$\alpha_{A,\ell}$ or $\alpha_{A,\ell}^{*}$. Hence, using
$\mathbb E\{\alpha_{A,\ell}\}=0$, their Bayesian averages satisfy
\begin{equation}
    \mathbb E
    \left\{
        \mathbf J_{\theta\alpha,\ell}
    \right\}
    =
    \mathbf 0_{1\times2}.
\end{equation}
Furthermore, the independence of $\theta_{AE,\ell}$ and
$\alpha_{A,\ell}$ implies that their prior information contains no
angle--amplitude cross terms. Consequently, the Bayesian Fisher
information matrix has the block-diagonal form
\begin{equation}
    \mathbf J_{B,\ell}
    =
    \begin{bmatrix}
        \operatorname{Tr}
        \left(
            \mathbf Q_{B,\ell}
            \mathbf R_x
        \right)
        +
        J_{P,\ell}
        &
        \mathbf 0_{1\times2}
        \\
        \mathbf 0_{2\times1}
        &
        \mathbf J_{B,\alpha\alpha,\ell}
    \end{bmatrix},
\end{equation}
where $\mathbf J_{B,\alpha\alpha,\ell}$ denotes the Bayesian
information block associated with the real and imaginary parts of the reflection coefficient.

Thus, the BCRB of the scalar angular parameter is the $(1,1)$ entry of
$\mathbf J_{B,\ell}^{-1}$. Since the Bayesian Fisher information
matrix is block diagonal,
\begin{equation}
    \mathrm{BCRB}_{\ell}(\mathbf R_x)
    =
    \left[
        \operatorname{Tr}
        \left(
            \mathbf Q_{B,\ell}
            \mathbf R_x
        \right)
        +
        J_{P,\ell}
    \right]^{-1},
\end{equation}
which proves Lemma~\ref{lem:eve_bcrb}.

\section{Proof of Proposition~\ref{prop:continuous_ghost_reformulation}}
\label{app:continuous_ghost_reformulation}

For arbitrary $\ell$, $i$, and $s\in\mathcal S_{\ell i}$, let
$d=N_e-1$. Using the ULA steering vector in (25), the reparameterized
quadratic form admits the finite Fourier expansion
\begin{equation}
    \widetilde{\mathbf a}_E^H(u)
    \mathbf C_{\ell}^{(i)}
    \widetilde{\mathbf a}_E(u)
    =
    \sum_{n=-d}^{d}
    c_{\ell i,n}e^{jnu},
    \label{eq:appendix_fourier_expansion}
\end{equation}
where, for $n=0,\ldots,d$,
\begin{equation}
    c_{\ell i,n}
    =
    \frac{1}{N_e}
    \sum_{m=1}^{N_e-n}
    \left[
        \mathbf C_{\ell}^{(i)}
    \right]_{m,m+n}.
    \label{eq:appendix_fourier_coefficients}
\end{equation}
Since $\mathbf C_{\ell}^{(i)}$ is Hermitian,
$c_{\ell i,-n}=c_{\ell i,n}^{*}$ for $n=1,\ldots,d$, while
$c_{\ell i,0}\in\mathbb R$. Thus, the quadratic form in
\eqref{eq:appendix_fourier_expansion} is real-valued, and
$q_{\ell i}(u)$ in \eqref{eq:response_difference_trig_poly} is a real
trigonometric polynomial of degree at most $d$.

Under the shifted tangent half-angle transformation in
\eqref{eq:tangent_half_angle},
\begin{equation}
    e^{ju}
    =
    e^{j\bar{u}_{\ell i,s}}
    \frac{1+j\tau}{1-j\tau}.
    \label{eq:appendix_tangent_exponential}
\end{equation}
Consequently, multiplication by $(1+\tau^2)^d$, with
$1+\tau^2=(1+j\tau)(1-j\tau)$, clears the denominators of all
trigonometric terms. Hence, $p_{\ell i,s}(\tau)$ in
\eqref{eq:ordinary_ghost_polynomial} is a real polynomial of degree
at most $2d$.

Since each spatial-frequency interval has width strictly smaller than
$2\pi$, its centered half-angle range lies strictly within
$(-\pi/2,\pi/2)$. The shifted tangent transformation is therefore
finite and bijective from $\mathcal U_{\ell i,s}$ onto
$[\tau_{\ell i,s}^{-},\tau_{\ell i,s}^{+}]$. Moreover,
$(1+\tau^2)^d>0$ for every real $\tau$. It follows that the
nonnegativity condition in
\eqref{eq:trig_polynomial_nonnegativity} is equivalent to the
nonnegativity of $p_{\ell i,s}(\tau)$ over
$[\tau_{\ell i,s}^{-},\tau_{\ell i,s}^{+}]$.

By the Markov--Luk\'acs theorem for univariate real polynomials, a
polynomial of degree at most $2d$ is nonnegative on
$[\tau_{\ell i,s}^{-},\tau_{\ell i,s}^{+}]$ if and only if it can be
expressed as the sum of an SOS polynomial of degree at most $2d$ and
\[
    \left(
        \tau-\tau_{\ell i,s}^{-}
    \right)
    \left(
        \tau_{\ell i,s}^{+}-\tau
    \right)
\]
multiplied by an SOS polynomial of degree at most $2(d-1)$. Using the
corresponding Gram-matrix representations with the monomial vectors
defined in Proposition~\ref{prop:continuous_ghost_reformulation}
yields \eqref{eq:interval_sos_identity}, where
$\mathbf Q_{\ell i,s}^{(0)}\succeq\mathbf0$ and
$\mathbf Q_{\ell i,s}^{(1)}\succeq\mathbf0$.

Finally, the coefficients of $q_{\ell i}(u)$, and consequently those of $p_{\ell i,s}(\tau)$, are affine in $\{\mathbf W_k\}_{k=1}^{K}$, $\mathbf Z$
and $\Delta$. Equating the coefficients of the corresponding powers of $\tau$ on both sides of \eqref{eq:interval_sos_identity} therefore produces affine equality constraints in the covariance variables, $\Delta$, and the Gram-matrix entries. This completes the proof.

\bibliographystyle{IEEEtran}
\bibliography{references}

\end{document}